\documentclass[12pt]{article}
\usepackage{setspace}
\usepackage[dvips]{graphics}
\usepackage{color}
\usepackage{epsfig}
\usepackage{amsfonts}
\usepackage{amsmath}
\usepackage{amssymb}
\usepackage{subfigure}

\newcommand{\nc}{\newcommand}
\def\sqr#1#2{{\vcenter{\vbox{\hrule
        height.#2pt \hbox{\vrule width.#2pt height#1pt \kern#2pt
          \vrule width.#2pt} \hrule height.#2pt}}}}

\newtheorem{theorem}{Theorem}[section]
\newtheorem{lemma}{Lemma}[section]

\newtheorem{corollary}{Corollary}[section]

\newtheorem{definition}{Definition}[section]

\newenvironment{proof}{{\sc Proof.}\hspace{3mm}}{\qed}
\nc{\qed}{\hfill $\blacksquare$}

\newcommand{\prob}{\mathbb{P}}

\def\Indic{1 \! \! 1}

\begin{document}

\title{Housing Market Microstructure\footnote{We would like to thank the seminar participants at Syracuse University, Cornell University and also thank the
conference participants for their helpful comments at INFORMS 2007 Annual Meeting, Bachelier Finance Society 2008, and AREUEA 2008.}}
\author{Hazer Inaltekin\thanks{
Princeton University, hinaltek@princeton.edu} \and Robert Jarrow\thanks{Cornell University, Johnson School of Management, raj15@cornell.edu} \and Mehmet Saglam
\thanks{Columbia University, Graduate School of Business, ms3760@columbia.edu} \and Yildiray Yildirim\thanks{Syracuse University, Martin J. Whitman School of Management, yildiray@syr.edu}}
\maketitle

\begin{abstract}
In this article, we develop a model for the evolution of real estate prices.
A wide range of inputs, including stochastic interest rates and changing
demands for the asset, are considered. Maximizing their expected utility,
home owners make optimal sale decisions given these changing market
conditions. Using these optimal sale decisions, we simulate the implied
evolution of housing prices providing insights into the recent subprime
lending crisis.
\end{abstract}


\begin{center}
\textit{\textbf{Keywords:}} Real estate market; Price evolution; Optimal
waiting time
\end{center}


\clearpage

\pagestyle{plain}

\section{Introduction}

The recent turmoil in financial markets triggered by defaults on subprime
mortgages demonstrates that home prices are largely driven by buyers'
demands. Indeed, in the early 2000s, house prices rose to record levels due
to the high demand for mortgages with easy credit. The left-side graph in
Figure \ref{fig:rate_inv} plots the Case-Shiller Composite-20 house-price
index, which shows that the prices peaked in 2006, were stable for a while,
and subsequently decreased. Figure \ref{fig:rate_inv} documents the level of
mortgage rates and home inventories during the same period. During the house
price appreciation cycle (2000-2005), there is a sharp decrease in mortgage
rates and a slight increase in home inventories. Low mortgage rates result
in an increased demand for mortgages and housing. In contrast, the period
with decreasing house prices (2006-2007) corresponds to a sharp increase in
mortgage rates and home inventories. In this paper, we propose a
mathematical model to explain the relation between house prices and buyers'
demands using a market microstructure perspective. We consider both market
and personal shocks, and we analyze how these shocks affect the home owner's
optimal sale decision and the resulting sale process. Our model can also be
used to forecast the future price evolution of house prices under different
interest rate and demand scenarios.

Our paper is related to two topics in the literature: the optimal waiting
time (OWT) and the home sale price evolution. Time-on-the-market (TOM),
time-to-sale, optimal marketing time, and selling time are frequently used
terms that have similar connotations to OWT even though their exact economic
definitions may differ \footnote{%
In our model, OWT is an upper bound for TOM.}. The existing literature in
this area is mostly empirical focusing on the sign of the correlation between
the length of TOM and the resulting sale price. Cubbin (1974) builds an
econometric model to explain the relationship between the list price and
selling time using data from the Coventry housing market between 1968 and
1970. Cubbin finds that the higher the list price, the shorter the selling
time. Miller (1978) cannot confirm Cubbin's conclusion using 83 sales from
Columbus, Ohio. Miller's empirical study does not confirm the existence of
an optimal selling time. Kalra, Chan and Lai (1997) analyze 644
single-family house sales records in the Fargo-Moorhead metropolitan area to
conclude that TOM and sale price are positively related. Genesove and Mayer
(1998) conclude from a study of the Boston condominium market that owners'
with a high loan-to-value ratio have a longer time on the market and sell
properties for higher prices. Taylor (1999) studies the theoretical
relationship between TOM and property quality. Knight (2002) examines how
changes in the list price impact the resulting TOM and sale price. In
contrast, our paper does not explicitly test the relationship between
expected TOM and our model's parameters, but rather we investigate the
relationship between OWT, list, and reservation prices. In addition, we also
examine the relationships between OWT and interest rates, and order arrival
and withdrawal intensities - which are not well-studied in the literature.

\begin{figure}[t]
\begin{center}
$%
\begin{array}{c@{\hspace{-0.5cm}}c@{\hspace{-0.5cm}}c}
\includegraphics[width=6cm]{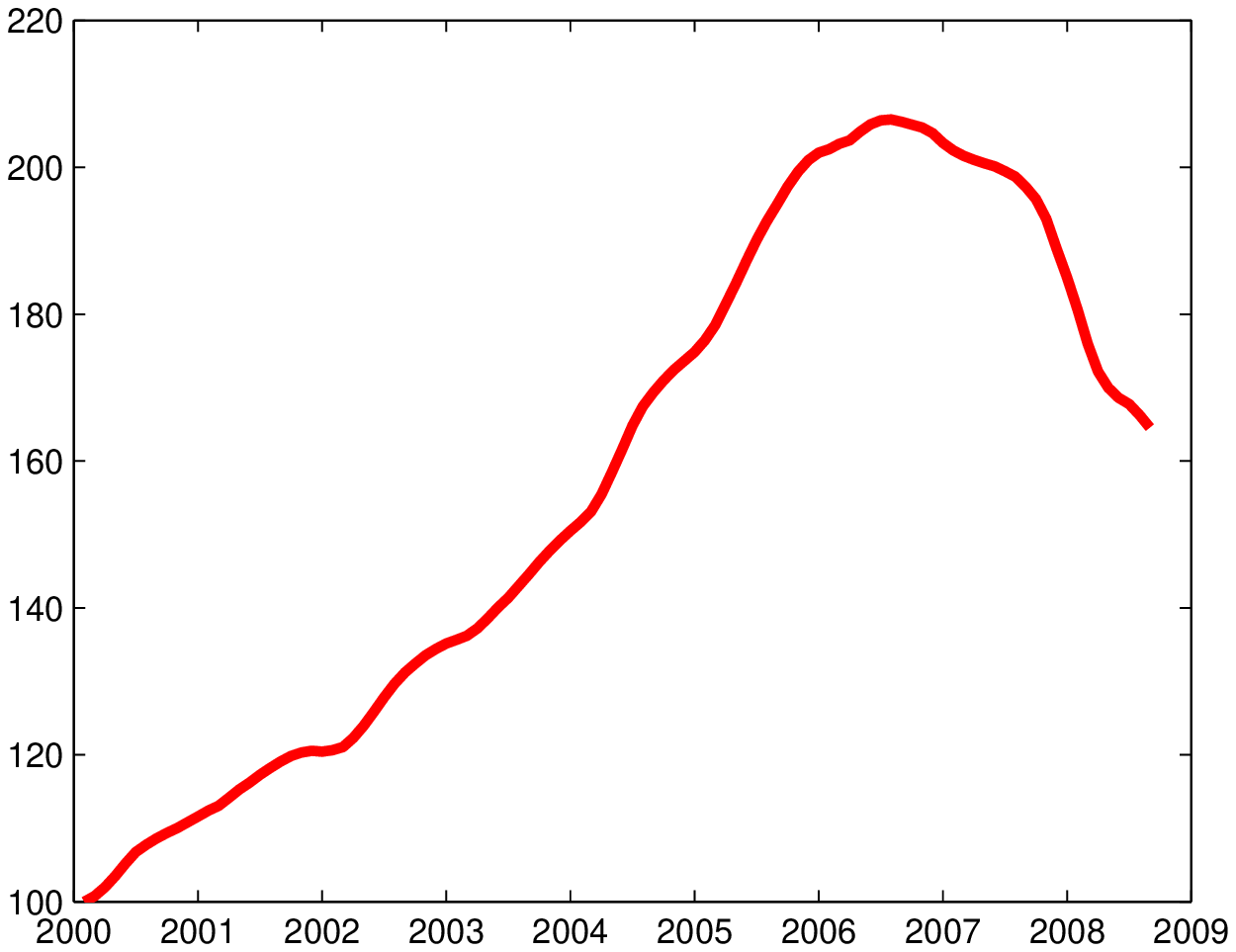} & %
\includegraphics[width=6cm]{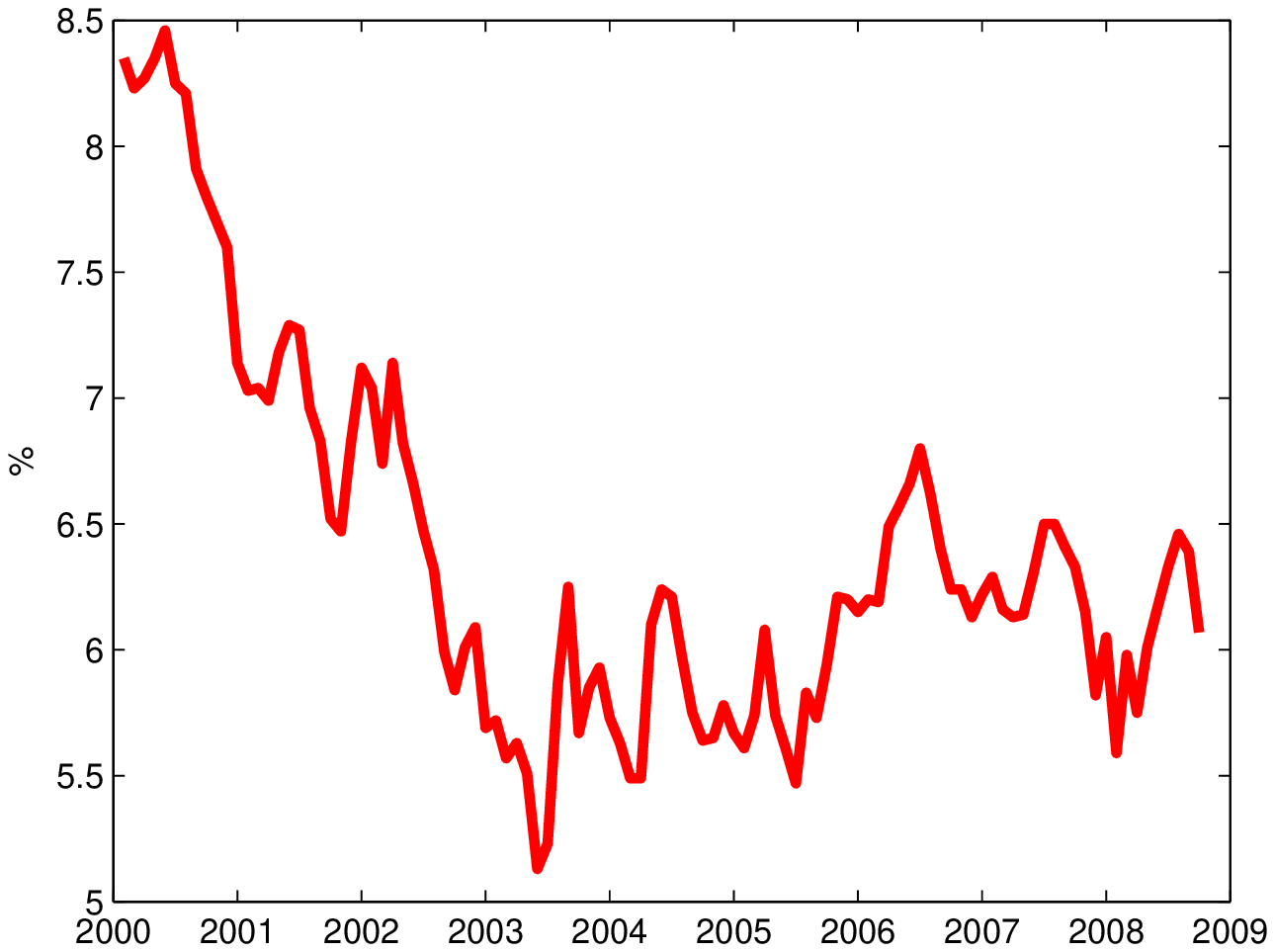} & %
\includegraphics[width=6cm]{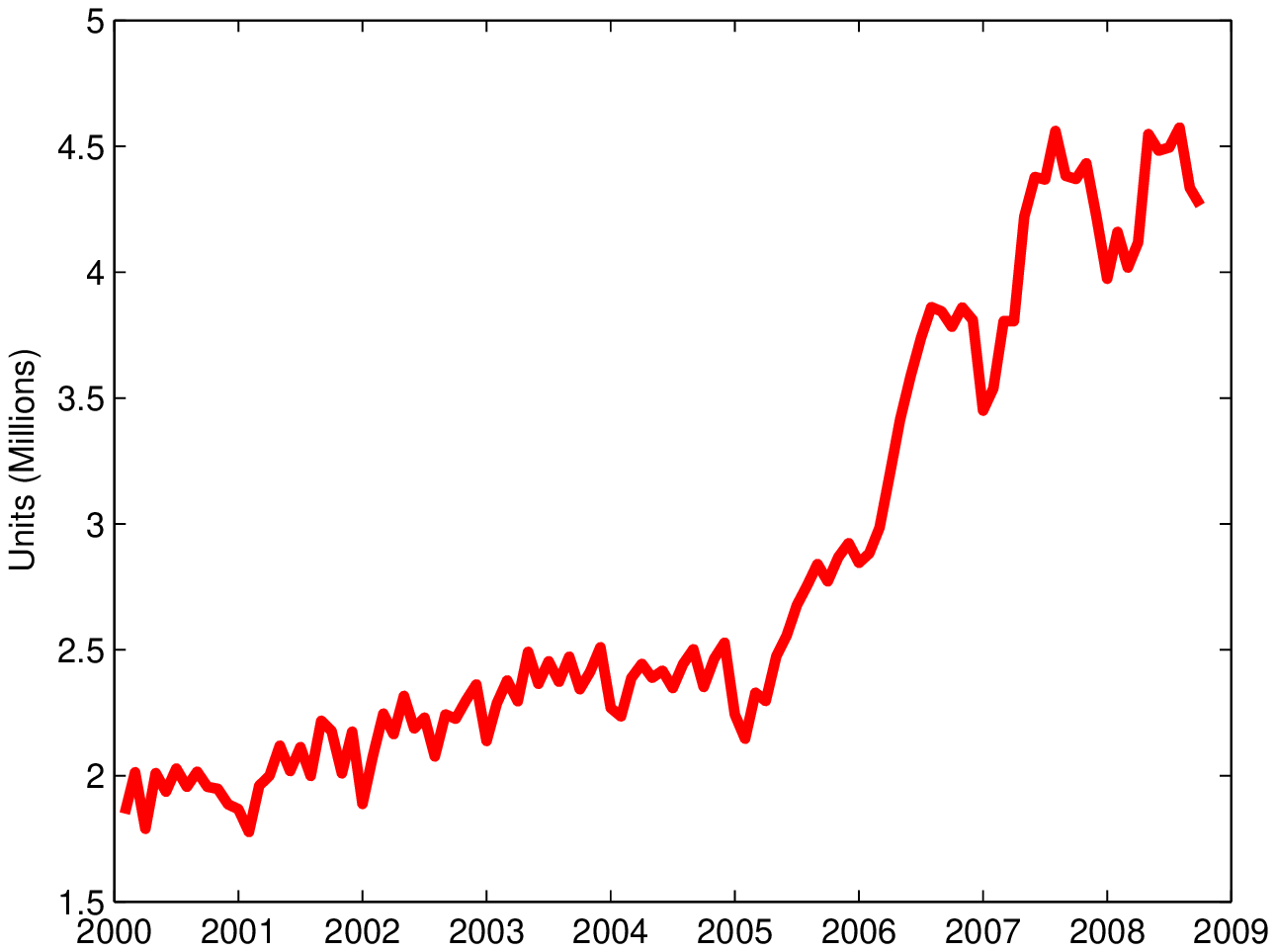} \\
\mbox{\footnotesize Case-Shiller Composite-20 index} &
\mbox{\footnotesize
30-year fixed mortgage rate} & \mbox{\footnotesize Existing home inventory}%
\end{array}%
$%
\end{center}
\caption{With the housing meltdown, the house-price index is falling,
existing inventory is piling up and current mortgage rate is significantly
higher than its trough.}
\label{fig:rate_inv}
\end{figure}


Many of the models studying TOM are based on information theory with search
and matching models. With their corresponding optimal stopping rules, these
models have been used to explain the behavior of buyers and sellers\footnote{%
Earlier search and matching models were used in labor economics, see Lucas
and Prescott (1974).}. Buyers search until the marginal benefit of continued
search is equal to marginal cost, and similarly, sellers equate the marginal
benefit to the marginal cost of locating a bidder for their property. Haurin
(1988) applies search theory to investigate the relationship between the
distribution of offers and the duration of marketing time. His empirical
study concludes that as the variance of the offer distribution increases,
the expected marketing time lengthens. Sarr (1988) examines the optimal list
price adjustment under demand uncertainty. Wheaton (1990) investigates the
role of vacancy rates in determining TOM, the reservation, and sale prices.
He finds that greater vacancy increases selling time, lowers the seller's
reservation price, and ultimately leads to lower market prices. Although we
do not study vacancy rates in our model, our findings are similiar to
Wheaton's (1990) for a seller with a low reservation price. In this case,
OWT increases and the seller expects less from the transaction. Forgey,
Rutherford and Springer incorporate liquidity into a search model. Using
data from $3358$ single-family housing transactions, they conclude that an
optimal marketing period exists and properties with higher liquidity sell at
a higher price. Yavas (1992), Krainer and LeRoy (2002) and Williams (1995)
also apply search and matching theory to analyze the sale prices of illiquid
assets. Our paper differs from these papers in our modeling approach. We do
not use a search and matching model, but instead employ a market
microstructure model. We incorporate the reservation price, the list price,
the distribution of offers, withdrawal rates for buyers' offers, and
deterministic and stochastic interest rates. The seller's motivation is
captured by a time impatient utility function\footnote{%
The existence of such a time impatience parameter is shown by Glower, Haurin
and Hendershott (1998).}. We find the optimal selling time that maximizes
the seller's expected utility. OWT in our model is set at the beginning of
the sales process. Optimal timing of investment has been previously studied
in the case of known asset price dynamics by Grenadier and Wang (2005) and
Evans, Henderson and Hobson (2007). In our paper, we do not study the risk
associated with the waiting period - which is well-documented by Lin and
Vandell (2007).

Real estate market price evolutions are considered in the second part of our
paper. The existing literature on price evolutions relies mostly on
econometric models. Since the housing market is heterogeneous in the
qualities of the properties, most of the existing literature is devoted to
developing statistical techniques to overcome this heterogeneity by forming
price indices for geographical areas, see Bailey, Muth and Nourse (1963),
Case K. and Shiller (1987), Case K. and Shiller (1990), Case B. and Quigley
(1991), and Poterba, Weill and Shiller (1991), Goetzmann and Peng (2006),
and McMillen and Thorsnes (2006). Another stream of literature uses
equilibrium theory to estimate house price dynamics. Stein (1995) explains
the large swings in prices by introducing an equilibrium model with a
down-payment effect. Similarly, Ortalo-Magne and Rady (2006) present a
recursive equilibrium model that accounts for income shocks and credit
constraints. Capozza, Hendershott and Mack (2004) use mean reversion and
serial correlation to explain equilibrium house price dynamics. Our model
differs from this literature. We use the derived analytics from our OWT
framework and extend it to multiple selling periods to determine the price
evolution. Given the linkage between the seller and the buyer (e.g., the
transaction price becomes the reservation price of the buyer when he posts
the property for sale), we construct a property price evolution by tracking
the expected sale price in the selling process. Our OWT framework enables us
to examine housing price movements under different interest rate and demand
scenarios. We model the timing of the decision of sale for each owner with
random income and personal shocks. Simulation shows that the fluctuation in
housing prices are driven by interest rates, demand for the asset, and
reservation prices.

This paper is organized as follow. Section 2 introduces the model within the
OWT framework. We start with an auxiliary model, extend it to more realistic
cases, and analyze the comparative statistics of OWT with respect to model
parameters. Section 3 applies our model to the price evolution in the real
estate market. Section 4 provides a simulation and reflects on the subprime
lending crisis. Finally, Section 5 concludes.

\section{The Model}

Sellers of illiquid real assets often face a difficult decision regarding
how long they should keep the asset on the market if they do not receive any
offers matching the list price. If the asset is highly desirable, the seller
might remain undecided even in the case of receiving an offer at the list
price. He may wait an additional amount of time in hopes of receiving an
offer even greater than the list price. If he chooses to wait longer, he may
lose the current offer. As this discussion implies, the determination of the
waiting time is complicated when there is a mutual decision process between
the seller and the buyer. This section analyzes the optimal amount of time
that a seller should wait in this mutual decision process to maximize his
expected payoff.

We consider two cases in our optimal waiting time (OWT) analysis. In the
first, the seller does not specify his final list price and accepts the
highest available offer exceeding his reservation price, $R$, at the end of
the OWT. In the second, he publicly announces the list price, $L$, and keeps
the reservation price $R$ private. He sells the asset if he receives an
offer greater than $L$; otherwise, he waits until the end of the OWT and
then chooses the best offer greater than $R$. We assume that buyers make
offers at random times with random magnitudes. Buyers may also withdraw
their offers according to a known random process.

We begin with the mathematical formulation of an auxiliary model. We provide
the details of the derivation for this auxiliary model in the appendix. We
then consider the two cases: (i) where the reservation and list prices are
private information, and (ii) where the seller announces a pre-determined
list price.

\subsection{The Auxiliary Model}

In the auxiliary model, the seller's reservation price, denoted $p_{\min }$,
is public information (i.e., all of the offers that the seller receives are
higher than $p_{\min }$). We assume that the arrival times of the buyers'
offers follow a one-dimensional Poisson point process with parameter $%
\lambda $, and the magnitudes of their offers are uniformly distributed with
$U(p_{\min },p_{\max })$ where $p_{\max }$ is finite. After making an offer,
a buyer may withdraw his offer. The time to withdrawal is assumed to follow
an exponential distribution with parameter $\mu $. Lastly, interest rates
are assumed constant and equal to $r$.

The seller maximizes their discounted expected payoff with respect to the
waiting time, $T$. Let the expected discounted payoff function be denoted by $
u(T,\lambda ,\mu ,p_{\min },p_{\max },r)$. The following lemma characterizes
this quantity.

\begin{lemma}
\begin{align}
u(T,\lambda ,\mu ,p_{\min },p_{\max },r)& =\left( -\frac{g(T,\lambda ,\mu
,r)(p_{\max }-p_{\min })}{(1-f(T,\mu ))^{2}}\right) \times  \notag \\
& \left( f(T,\mu )e^{\lambda Tf(T,\mu )}-\frac{1}{\lambda T}(e^{\lambda
Tf(T,\mu )}-1)-e^{\lambda Tf(T,\mu )}+\frac{1}{\lambda T}(e^{\lambda
T}-1)\right)  \notag \\
&+\frac{p_{\max }g(T,\lambda ,\mu ,r)(e^{\lambda T}-e^{\lambda Tf(T,\mu )})}{%
1-f(T,\mu )}
\end{align}%
where $f(T,\mu )=1-\frac{1}{\mu T}(1-e^{-\mu T})$ and $g(T,\lambda ,\mu
,r)=(1-f(T,\mu ))e^{-rT}e^{-\lambda T}$.
\end{lemma}

\begin{proof}
See Appendix A.
\end{proof}

This lemma is used in the following two cases.

\subsection{No List Price}

In this case the seller's reservation price, $R$, is private information.
The seller does not post a list price, and considers all offers until the
end of the waiting time. Here the seller has a minimum price for a sale, but
doesn't limit the upside payoff.

Using the thinning principle for Poisson processes, Resnick (1992), the
expected discounted payoff, $v(T,\lambda ,\mu ,R,p_{\min },p_{\max },r)$,
can be computed as in the following corollary.

\begin{corollary}
\label{Corolloary:case2} $v(T,\lambda,\mu,R,p_{\min},p_{\max},r) =
u(T,\lambda_{\text{thinned}},\mu,R,p_{\max},r)$ where $\lambda_{\text{thinned%
}} = \lambda\frac{p_{\max}-R}{p_{\max}-p_{\min}}$.
\end{corollary}

\begin{figure}[!t]
\begin{center}
\includegraphics[width=11cm]{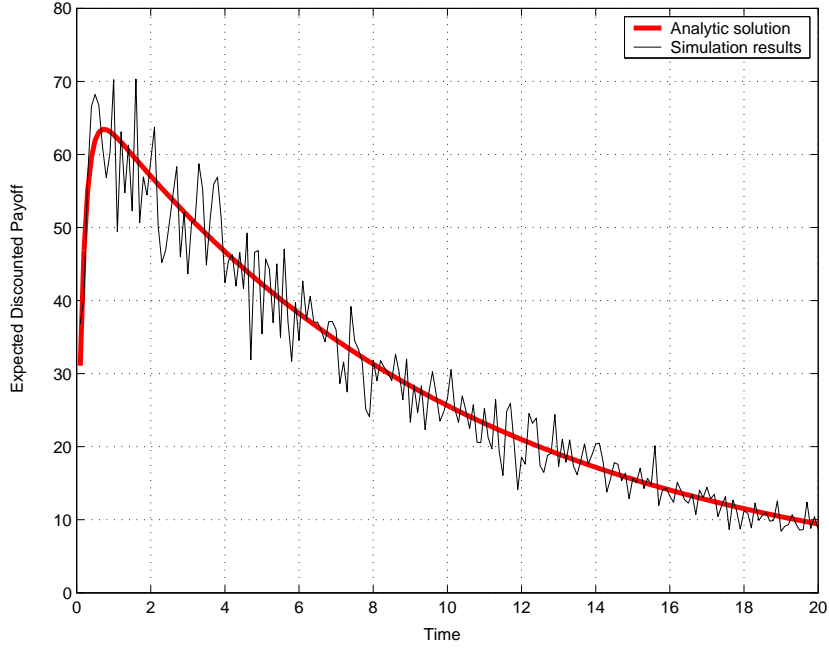} 
\end{center}
\caption{Expected discounted payoff in the case without a list price.}
\label{fig:case2}
\end{figure}

Figure \ref{fig:case2} plots the expected discounted payoff as in Corollary %
\ref{Corolloary:case2}. The values of the parameters used in this figure are
contained in the appendix. The expected discounted payoff equals zero at $T=0
$, increases to a maximum value, and then starts to decrease. This plot
shows that there is an optimal waiting time.

\subsection{Public List Price}

\label{section:case2}

In this case the seller announces a list price, $L$, and still has a private
reservation price, $R$. This case can be simplified by dividing the payoff
function into two regions. From Corollary \ref{Corolloary:case2}, if the
seller does not receive any offers higher than $L$, the payoff becomes $%
u(T,\lambda \frac{L-R}{p_{\max }-p_{\min }},\mu ,R,L,r)$. The remaining
region is when offers higher than $L$ are received. The payoff function, $%
w(T,\lambda ,\mu ,R,L,p_{\min },p_{\max },r)$, is obtained by considering
these two regions.

\begin{theorem}
\label{theorem:case3}
\begin{align}
w(T,\lambda ,\mu ,R,L,p_{\min },p_{\max },r)& =\left( 1-e^{-\lambda
Ty}\right) \left( \frac{p_{\max }+L}{2}\right) \left( \frac{\lambda y}{%
\lambda y+r}\right)  \notag \\
& \qquad +e^{-\lambda Ty}u(T,\lambda x,\mu ,R,L,r)
\end{align}%
where $x=\frac{L-R}{p_{\max }-p_{\min }}$ and $y=\frac{p_{\max }-L}{%
p_{\max}-p_{\min }}$.
\end{theorem}

\begin{proof}
If the seller does not receive any offer higher than $L$, the payoff equals $%
u(T,\lambda \ x,\mu ,R,L,r)$. Given that there is an offer higher than $L$,
our payoff equals $\frac{p_{\max }+L}{2}\mathbb{E}\left[ e^{-r\beta }\right]
$ where $\beta $ is a random variable representing the arrival time of the
first offer greater than $L$. If $y=\frac{p_{\max}-L}{p_{\max }-p_{\min }}$,
then $\mathbb{E}\left[ e^{-r\beta }\right] $ equals the moment-generating
function of an exponential random variable with parameter $\lambda y$. Thus,
$\mathbb{E}\left[ e^{-r\beta }\right] =\frac{\lambda y}{\lambda y+r}$. \ To
find the resulting expected payoff, we only need the probability of
receiving an offer greater than $L$, which equals $\left( 1-e^{-\lambda
Ty}\right) $. As a result, our expected payoff function is the sum of these
two parts multiplied by their corresponding probabilities.
\end{proof}

\begin{figure}[!t]
\begin{center}
\includegraphics[width=11cm]{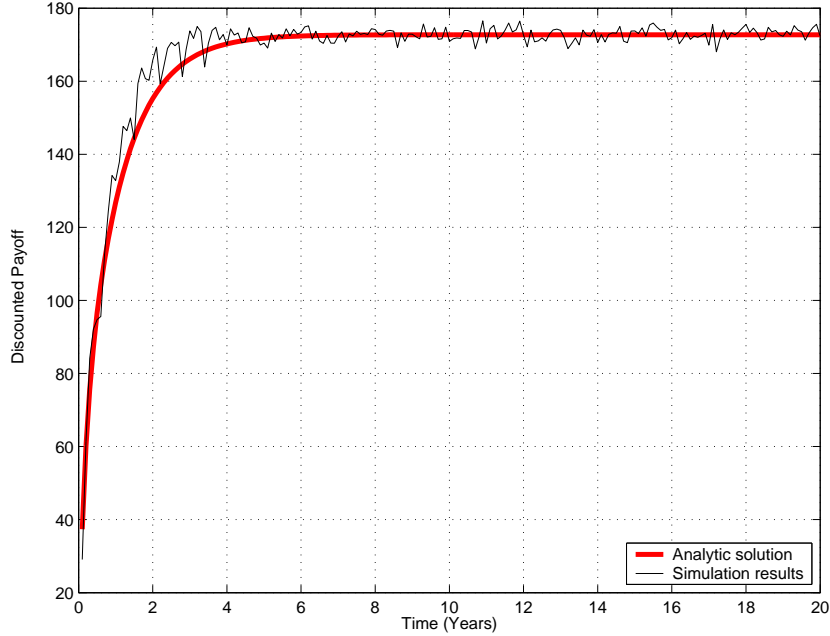} 
\end{center}
\caption{Expected discounted payoff in the case with an announced list price}
\label{fig:case3}
\end{figure}

Figure \ref{fig:case3} graphs the expected discounted payoff with respect to
the waiting time for a public list price. The values of the parameters used
in this graph are contained in the appendix. Note that the expected payoff
function is strictly increasing in $T$ with an asymptote:
\begin{equation}
\lim_{T\rightarrow \infty }w(T,\lambda ,\mu ,R,L,p_{\min },p_{\max
},r)=\left( \frac{p_{\max }+L}{2}\right) \left( \frac{\lambda y}{\lambda y+r}%
\right) .
\end{equation}%
When $T$ increases, the payoff function is dominated by the region where the
seller receives an offer greater than $L$. Unlike the first model, the
asymptotic payoff is not a function of the waiting time, $T$, and therefore,
does not diminish with respect to $T$. By setting a longer waiting time, the
seller can almost surely get an offer greater than $L$ (and because the
discount factor, $\mathbb{E}\left[ e^{-r\beta }\right] $, does not depend on
$T$.) This is due to the properties of the exponential distribution for the
first offer greater than $L$. This figure does not imply that the seller
waits an infinite amount of time to sell the asset. Rather, it only implies
that the seller makes a conservative estimate of the maximum waiting time
before the sale process starts.

The reason why the expected discounted payoff is increasing in $T$ is
because there is no specific utility function associated with the seller's
motivation to sell the asset (see Glower, Haurin and Hendershott (1998) who
show that the seller's motivation is a significant factor in determining the
selling time and the sale price). We can introduce this utility based sales
motivation via a utility function, $\mathcal{U(.)}$, defined as follows.

\begin{definition}
$\mathcal{U}(D,T)=De^{-\gamma T}$ where $D$ is the discounted payoff, $T$ is
the waiting time, and $\gamma \geq 0$ is a time impatience parameter.
\end{definition}

As the time impatient parameter $\gamma $ increases, the seller is more
motivated to sell the asset. This \emph{selling motivation} enables the
model to incorporate different individuals, who have different sale
processes even under the same market conditions. With this assumption, the
expected utility function, $z(.)$, can be written as follows.

\begin{corollary}
\label{theorem:case4}
\begin{align}
z(T,\lambda,\mu,R,L,p_{\min},p_{\max},r,\gamma) &= \mathbb{E}\Big[\mathcal{%
U(.)}\Big] \\
& = e^{-\gamma T}w(T,\lambda,\mu,R,L,p_{\min},p_{\max},r)
\end{align}
\end{corollary}

\begin{figure}[!t]
\begin{center}
\includegraphics[width=11cm]{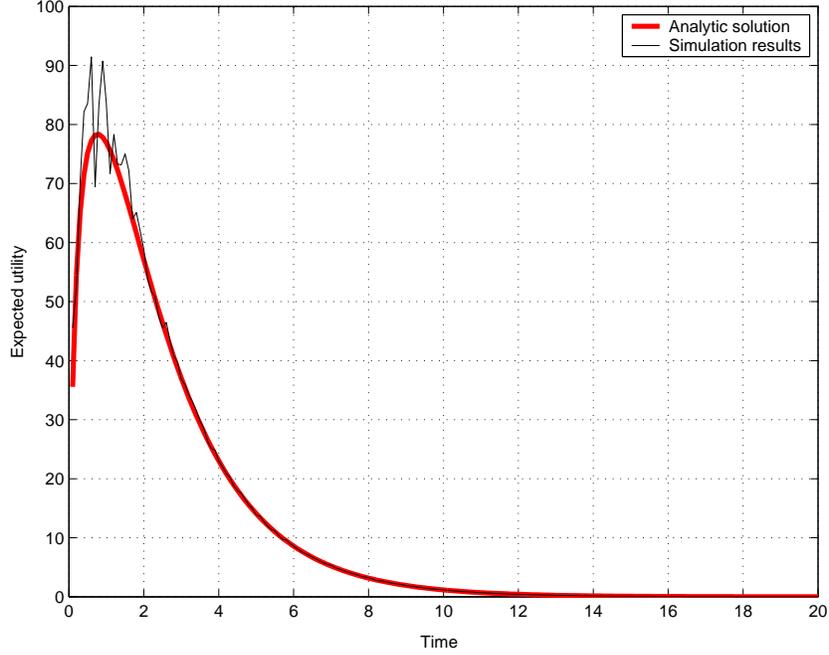} 
\end{center}
\caption{Expected utility in the case of an announced list price}
\label{fig:case3b}
\end{figure}

Figure \ref{fig:case3b} plots the expected utility with respect to the
waiting time in the case of a public list price. The parameter values used
in this graph are contained in the appendix.

\subsection{An Analysis of OWT}

This section defines the OWT and analyzes its comparative statistics with
respect to the model's parameters. We use the public list price case of
Section \ref{section:case2} as our underlying model.

\begin{definition}
Let $T^{\ast }$ denote the OWT that maximizes expected utility. Then,
\begin{equation}
T^{\ast }=\arg {\max {\{T\geq 0:z(T,\lambda ,\mu ,R,L,p_{\min },p_{\max
},r,\gamma )\}}}.
\end{equation}
\end{definition}

Since the expected utility function is continuous and its first derivative
has a unique sign change from positive to negative, a global maximum exists.
We plot $T^{\ast }$ as a function of model's parameters in Figures \ref%
{fig:combined_lambda_mu_r} through \ref{fig:combined_S_R_r}. The parameter
values used in the graphs are contained in the appendix.

\begin{figure}[!t]
\begin{center}
$%
\begin{array}{c@{\hspace{.5cm}}c}
\includegraphics[width=8cm]{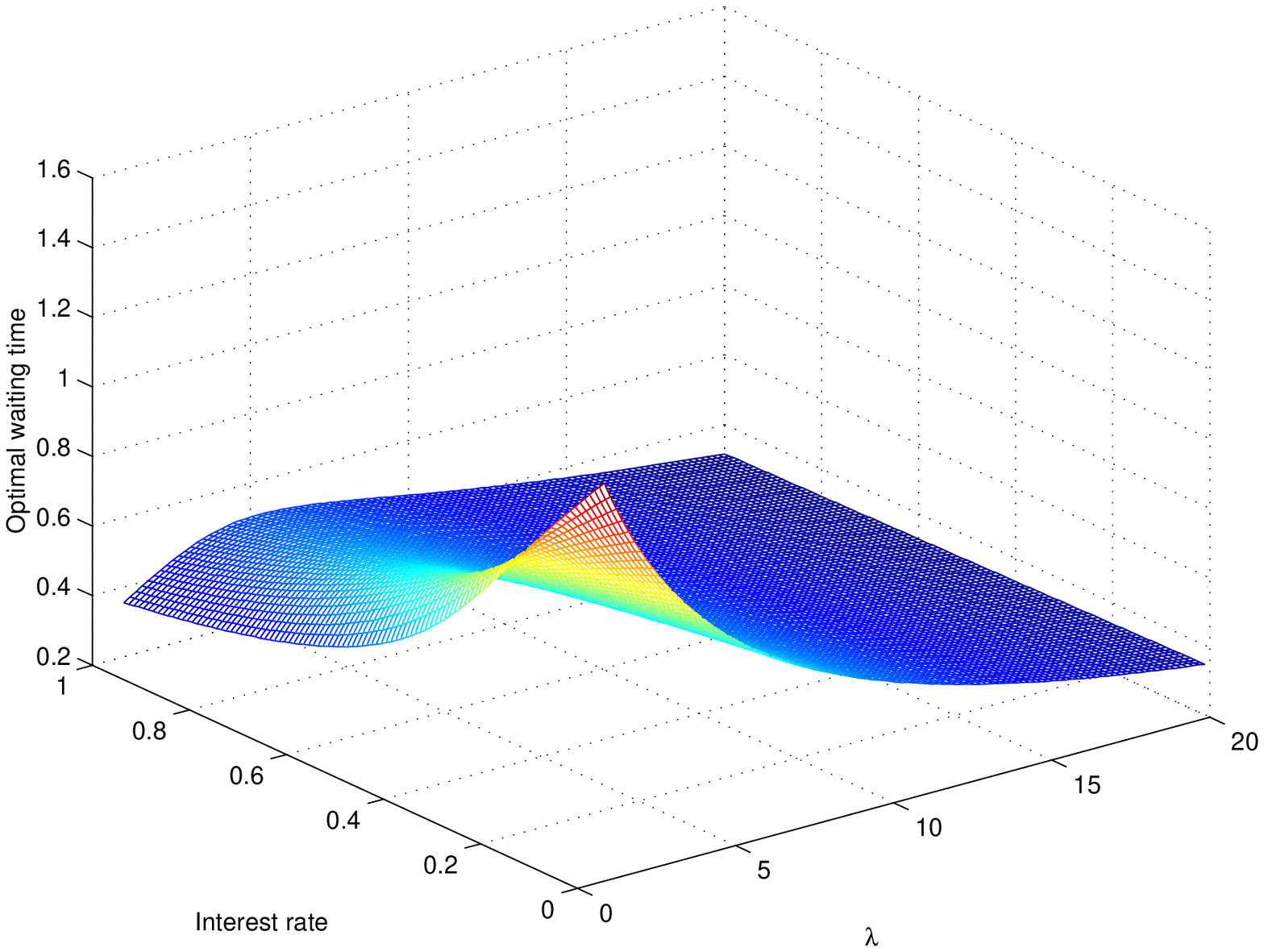} & %
\includegraphics[width=8cm]{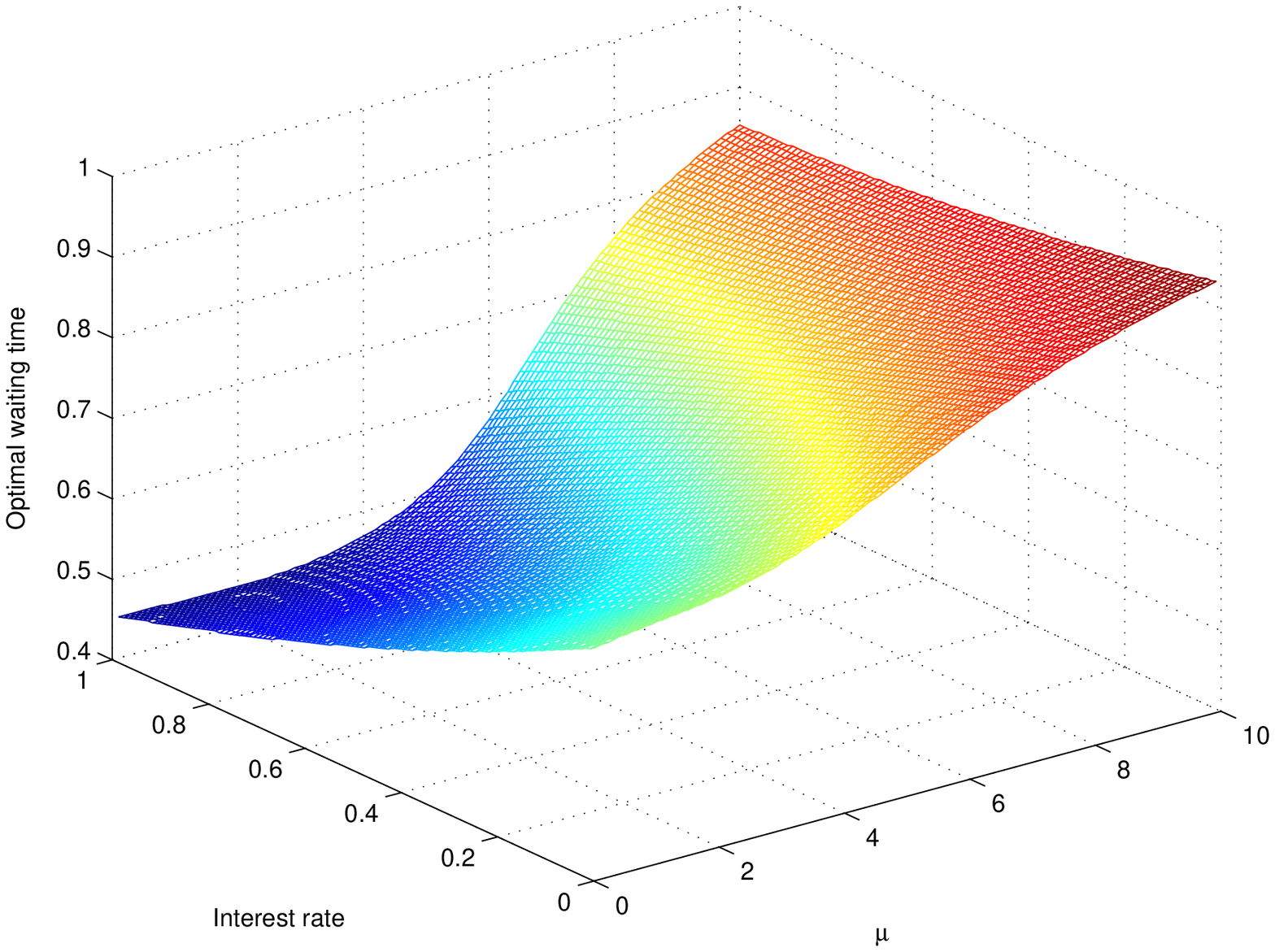} \\
&
\end{array}%
$%
\end{center}
\caption{Optimal waiting time as a function of arrival intensity and
interest rate (left) and optimal waiting time as a function of withdrawal
intensity and interest rate (right)}
\label{fig:combined_lambda_mu_r}
\end{figure}

The graph on the left side of Figure \ref{fig:combined_lambda_mu_r} shows
the change in $T^{\ast }$ with respect to the arrival intensity of offers, $%
\lambda $, and interest rates $r$. As $\lambda $ increases, the seller
chooses a smaller waiting time because he expects to get sufficiently many
offers and, thus, he avoids a large utility loss from waiting longer. This
graph illustrates this point with a decreasing concave-up function when
interest rates are low. However, when interest rates are high, we see that $%
T^{\ast }$ actually increases when $\lambda $ is small and increasing. This
implies that in this region, the amount of additional offers created by
higher demand is worth the wait.

The graph on the right-side graph of Figure \ref{fig:combined_lambda_mu_r}
shows the change in $T^{\ast }$ with respect to the withdrawal intensity of
the buyers, $\mu $, and interest rates $r$. As $\mu $ increases, the seller
chooses a longer waiting time because he wants to increase the number of
offers, as in the case of the increasing arrival intensity. Figure \ref%
{fig:combined_lambda_mu_r} shows that, when interest rates are high, there
is a constant optimal waiting time. In this circumstance, since the discount
rate is high, small changes in the withdrawal intensity do not affect $%
T^{\ast }$.

Both of these graphs show that interest rates are indirectly proportional
with $T^{\ast }$. This is intuitive because the seller does not want to wait
long in a high interest rate environment.

\begin{figure}[tbp]
\begin{center}
\includegraphics[width=11cm]{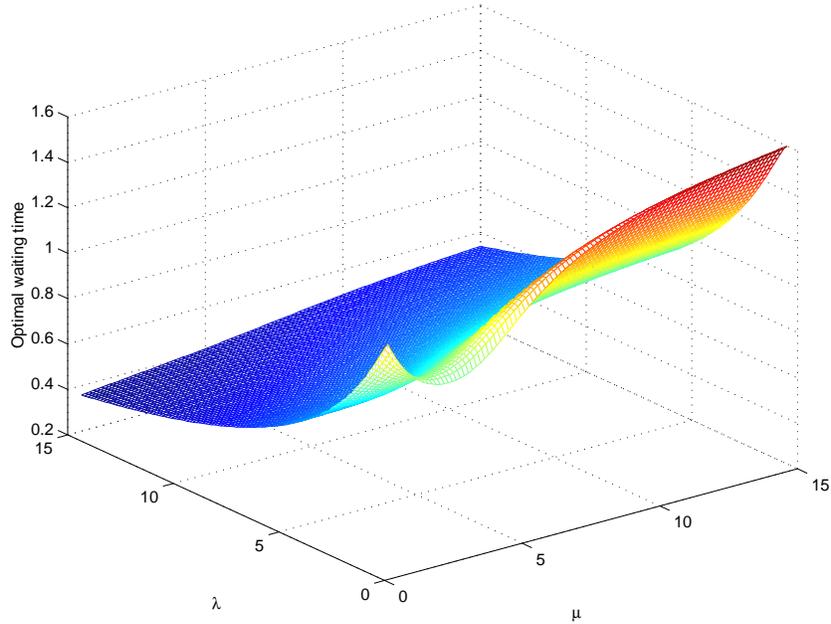} 
\end{center}
\caption{Optimal waiting time as a function of arrival intensity and
withdrawal intensity}
\label{fig:lambda_mu}
\end{figure}

Figure \ref{fig:lambda_mu} illustrates that the intensities of offer
arrivals $\lambda $ and offer withdrawals $\mu $ do not always create the
opposite impact on $T^{\ast }$. For small values of $\lambda $ and $\mu $, $%
T^{\ast }$ increases when they both rise above a certain threshold. Only
after this threshold do they begin to create the opposite effect.

The left-side of Figure \ref{fig:combined_S_R_r} shows the change in $%
T^{\ast }$ with respect to the list price and interest rates. When the list
price increases, the OWT increases. However, as Figure \ref%
{fig:combined_S_R_r} illustrates, if our list price is already high,
increasing the list price further diminishes the probability of receiving an
offer greater than this new value and, thus, there is no incentive to wait
longer.

The right-side of Figure \ref{fig:combined_S_R_r} shows the change in $%
T^{\ast }$ with respect to the reservation price and interest rates. As $R$
increases, the seller waits longer because, using the thinning principle in
Poisson processes, this case implies a decreasing $\lambda $. As we have
seen earlier, decreasing the arrival intensity results in a larger $T^{\ast }
$.

\begin{figure}[tbp]
\begin{center}
$%
\begin{array}{c@{\hspace{.5cm}}c}
\includegraphics[width=8cm]{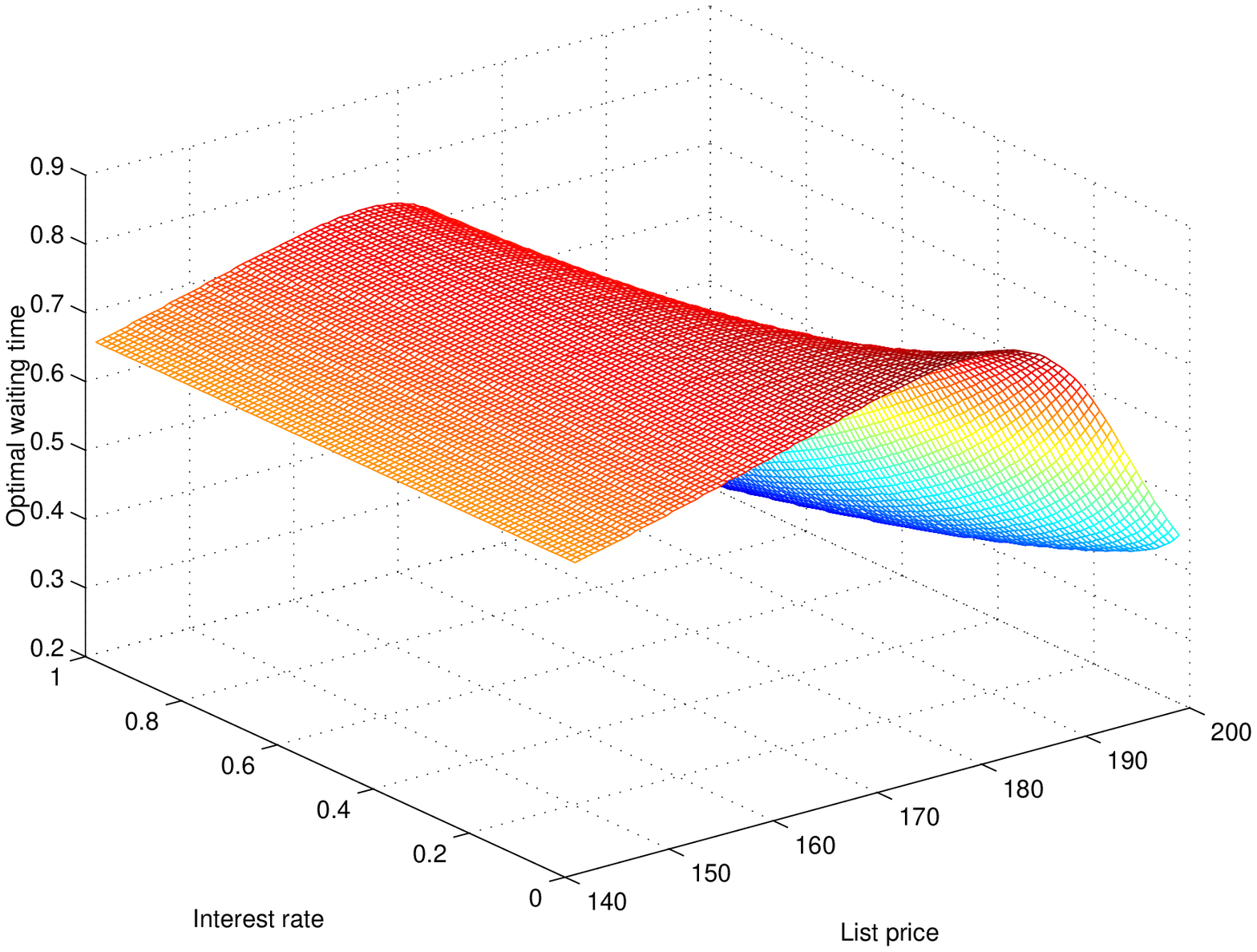} & %
\includegraphics[width=8cm]{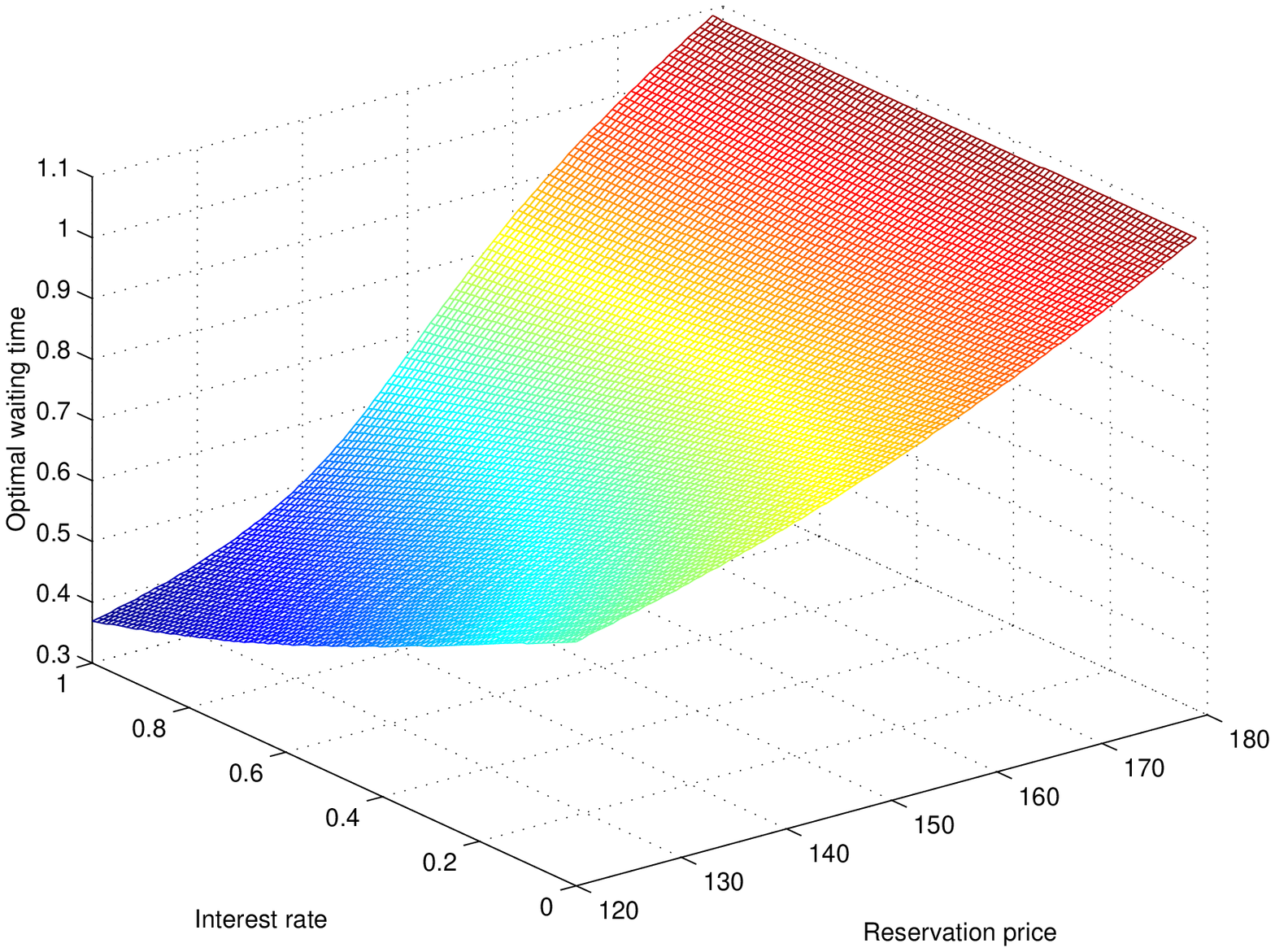} \\
&
\end{array}%
$%
\end{center}
\caption{Optimal waiting time as a function of list price and interest rate
(left) and optimal waiting time as a function of reservation price and
interest rate (right)}
\label{fig:combined_S_R_r}
\end{figure}

\section{Modeling Real Estate Price Evolutions}

The previous sections of the paper characterized the optimal waiting time
when selling an illiquid asset. The remainder of the paper studies the time
evolution of real estate sale prices in a market where sellers maximize
their payoffs by considering the optimal waiting time.

\subsection{Stochastic Demand and Interest Rates}

\label{subsection:macro}

The real estate market price evolution is strongly affected by changes in
the broader economy such as recessionary or expansionary cycles and by
interest rates shocks. In our model, we model the recessionary and
expansionary cycles with a stochastic demand, $\lambda (t)$. The demand, $%
\lambda (t)$, will be a function of interest rates, $r(t)$, and the
announced list prices, $L(t)$, i.e.
\begin{equation}
\lambda (t)=g\big(r(t),L(t)\big)
\end{equation}%
where $r(t)$ is stochastic and $L(t)$ is non-random for the seller.

It is assumed that $L(t)$ is non-increasing with $L(t)\geq R>0$ for all $%
t\geq 0$, where $R$ is the seller's reservation price. Offers arrive
according to a non-homogenous Poisson process with stochastic arrival
intensity, $\lambda (t)$. Let $\mathcal{F}_{t}$ be the $\sigma $-algebra
generated by $\{r(s)\}_{0\leq s\leq t}$, and $N(t)$ be the number of
arrivals in $[0,t]$. Then,

\begin{eqnarray}
\mathbb{P}\{N(t)=n\} &=&\mathbb{E}\Big[\Lambda (t)^{n}\frac{\exp (-\Lambda
(t))}{n!}\Big],\mbox{ and}  \notag \\
\mathbb{P}\{N(t)=n|\mathcal{F}_{t}\} &=&\Lambda (t)^{n}\frac{\exp (-\Lambda
(t))}{n!}
\end{eqnarray}%
where $\Lambda (t)=\int_{0}^{t}\lambda (s)$\textrm{d}$s$.

Offers' values come from an independent distribution, $F_{\xi }$ with density
$f_{\xi }$, where $\xi $ represents the intensity of a generic offer. $\xi
_{i}$ and $A_{i}$, $i\geq 1$, represent the intensity and the arrival time
of the $i$th offer, respectively. Given $\mathcal{F}_{t}$ and the number of
offers in $[0,t]$, offer arrival times are independently distributed over $%
[0,t]$ with conditional density
\begin{equation*}
f_{A|\{\mathcal{F}_{t},N(t)=n\}}(a)=\frac{\lambda (a)}{\Lambda (t)}.
\end{equation*}

After an offer arrives, it is withdrawn after a random time with continuous
distribution function $F_{\tau }$, where $\tau $ is the withdrawal time. It
is not hard, however, to extend the analysis to more general offer
withdrawal time distributions. We let $\tau _{i}$, for $i\geq 1$, represents
the withdrawal time of the $i$th offer. All $\tau _{i}$'s are assumed to be
independent from one another.

Without loss of generality, we set $\xi _{0}=A_{0}=\tau _{0}=0$. Let $\beta $
be the arrival time of the first offer greater than the list price. Then,
\begin{equation*}
\beta =\inf \{A_{i}:\xi _{i}\geq L(A_{i}),i\geq 1\}.
\end{equation*}

We prove three theorems characterizing the expected payoff to the seller. In
the first case, the seller announces a time-dependent list price. In this
case, the asset is sold whenever there is an offer greater than the ask
price. This is the most general case containing the others as subcases. In
the second case, the list price of the asset does not change with time.
Finally, in the third case, the seller does not announce a list price. He
waits an optimal amount of time and then chooses the best available offer
greater than the reservation price.

For the first case, the discounted payoff function, $X(t)$, at time $t$ is:
\begin{eqnarray}
X(t) &=&\exp \Big(-\int_{0}^{t}r(s)ds\Big)\cdot \Big(\max_{0\leq i\leq
N(t)}\xi _{i}\cdot 1\!\!1_{\{\xi _{i}\geq R\}}\cdot 1\!\!1_{\{\tau _{i}\geq
t-A_{i}\}}\Big)\cdot 1\!\!1_{\{\beta >t\}}  \notag \\
&+&\exp \Big(-\int_{0}^{\beta }r(s)ds\Big)\xi _{i(\beta )}1\!\!1_{\{\beta
\leq t\}}  \label{Eqn: Discounted Payoff}
\end{eqnarray}%
where $\xi _{i(\beta )}$ is the offer value at $\beta $. The first term in (%
\ref{Eqn: Discounted Payoff}) accounts for the situation where all offers
until time $t$ are smaller than the list price. The second term corresponds
to the situation where there is an offer price greater than the list price
before time $t$. We have the following theorem for the expected discounted
payoff at time $t$.

\begin{theorem}
\label{Thm: Payoff with changing Sale price} Let $P(t) = \mathbb{E}[X(t)]$
and $P(t | \mathcal{F}_t) = \mathbb{E}[X(t) | \mathcal{F}_t]$. Then, $P(t) =
\mathbb{E}[P(t | \mathcal{F}_t)]$, and
\begin{eqnarray}
P(t | \mathcal{F}_t) &=& \mathrm{e}^{-\int_0^t r(s)\mathrm{d}s} \mathrm{e}%
^{\Lambda(t)(\varphi(t)-1)}\Big(L_0 - \int_0^{L_0} \mathrm{e}^{-\Lambda(t)
\psi(t,y)} \mathrm{d}y \Big)  \notag \\
&+& \big(1-\mathrm{e}^{\Lambda(t)(\varphi(t)-1)}\big) \int_0^t
\int_{L(a)}^\infty \frac{\lambda(a)\exp\big(-\int_0^a r(s)\mathrm{d}s\big)}{%
\Lambda(t)\big(1-F_\xi\big(L(a)\big)\big)} f_\xi(x) x \mathrm{d}x \mathrm{d}a
\end{eqnarray}
where $\psi(t,y) = \frac{1}{\Lambda(t)} \int_0^t \lambda(a) \Big(%
1-F_\tau(t-a)\Big)\Big(F_\xi\big(L(a)\vee y\big) - F_\xi\big(R \vee y\big)%
\Big) \mathrm{d}a$ and $\varphi(t) = \frac{1}{\Lambda(t)} \int_0^t
\lambda(a) F_\xi\big(L(a)\big) \mathrm{d}a$.
\end{theorem}

\begin{proof}
See Appendix B.
\end{proof}

The difficulty in proving Theorem \ref{Thm: Payoff with changing Sale price}
is that a changing list price introduces a coupling between the offer
intensity and the arrival time.

In the second case, we assume that $L(t)$ is a constant. Here, the seller's
payoff can also be written exactly as in (\ref{Eqn: Discounted Payoff}).
Theorem \ref{Thm: constant list price} gives us the seller's expected payoff
at time $t$ for this case. We do not provide a proof because it is similar
to that of Theorem \ref{Thm: Payoff with changing Sale price}.

\begin{theorem}
\label{Thm: constant list price} Let $P(t) = \mathbb{E}[X(t)]$ and $P(t |
\mathcal{F}_t) = \mathbb{E}[X(t) | \mathcal{F}_t]$. If $L(t)=L\geq R>0$ is
constant, then $P(t) = \mathbb{E}\big[P(t | \mathcal{F}_t)\big]$ and
\begin{eqnarray}
P(t | \mathcal{F}_t) &=& \mathrm{e}^{-\int_0^t r(s)\mathrm{d}s } \mathrm{e}%
^{\Lambda(t)(F_\xi(L)-1)}\Big(L - \int_0^{L} \mathrm{e}^{-\Lambda(t)
\psi(t,y)} \mathrm{d}y \Big)  \notag \\
&+& \big(1-\mathrm{e}^{\Lambda(t)(F_\xi(L)-1)}\big) \int_0^t
\int_{L(a)}^\infty \frac{\lambda(a)\exp\big(-\int_0^a r(s)\mathrm{d}s\big)}{%
\Lambda(t)\big(1-F_\xi\big(L(a)\big)\big)} f_\xi(x) x \mathrm{d}x \mathrm{d}a
\end{eqnarray}
where $\psi(t,y) = \big(F_\xi(L \vee y) - F_\xi(R \vee y)\big)
\big(1-\frac{1}{\Lambda(t)} \int_0^t \lambda(a) F_\tau(t-a) \mathrm{d}a
\big)
$.
\end{theorem}

Finally, in the third case, the seller does not announce a list price. In
this case, his payoff process at time $t$ can be written as
\begin{equation}
X(t)=\exp \Big(-\int_{0}^{t}r(s)ds\Big)\cdot \Big(\max_{0\leq i\leq N(t)}\xi
_{i}\cdot 1\!\!1_{\{\xi _{i}\geq R\}}\cdot 1\!\!1_{\{\tau _{i}\geq t-A_{i}\}}%
\Big).  \label{Eqn: Discounted Payoff No Sale Price}
\end{equation}

The simplicity of (\ref{Eqn: Discounted Payoff No Sale Price}) compared to (%
\ref{Eqn: Discounted Payoff}) results in a simplification in the expected
payoff formula at time $t$, which is given in Theorem \ref{Thm: no sale
price}. Again, we do not provide a proof because it is similar to the proof
of Theorem \ref{Thm: Payoff with changing Sale price}.

\begin{theorem}
\label{Thm: no sale price} Let $P(t) = \mathbb{E}\left[X(t)\right]$ and $%
P\left(t | \mathcal{F}_t\right) = \mathbb{E}\left[X(t) | \mathcal{F}_t\right]
$. If the seller does not announce any list price (i.e., $L=\infty$), then $%
P(t) = \mathbb{E}\left[P(t | \mathcal{F}_t)\right]$ and
\begin{eqnarray}
P\left(t | \mathcal{F}_t\right) = \exp\left(-\int_0^t r(s)\mathrm{d}s
\right) \int_0^\infty \left(1-\exp\left(-\Lambda(t)\psi(t,y) \right)\right)
\mathrm{d}y,
\end{eqnarray}
where $\psi(t,y) = \left(1- F_\xi(R \vee y)\right) \left(1-\frac{1}{%
\Lambda(t)} \int_0^t \lambda(a) F_\tau(t-a) \mathrm{d}a \right)$.
\end{theorem}

Further simplifications to Theorems \ref{Thm: Payoff with changing Sale
price}, \ref{Thm: constant list price} and \ref{Thm: no sale price} are
possible depending on the distribution of the offer intensities, the
distribution of offer waiting times, and the distribution of the stochastic
process governing the offer arrival times.

\subsection{The Microstructure of Housing Markets}

This section uses the previous model via simulation to show how the list
price of a house changes over time, if the seller uses an optimal sale time.
We focus on a single house and track its price evolution for a given period
of time. During this horizon, the house may be sold a number of times and
the resulting sequence of sale prices constitutes the price evolution. For
each owner of the house, the model evolves similarly.

\subsubsection{The Occupation Period}

A new owner's occupation starts after the sale of the house. The buyer (new
owner) knows how much he paid for the asset and this constitutes his
reservation price, $R$. The owner will try to sell the asset for at least
this amount when he posts the asset for sale. We assume that each owner
lives in the house for a certain amount of time, $O$ years, which may vary
from individual to individual. Within this period, he does not want to sell
the asset unless there is a shock (such as relocation necessity,
unsuitability of the asset after a change in the size of his family,
personal insolvency or bankruptcy, default and foreclosure, etc.). If the
owner of the house does not realize a shock, he will not try to sell the
asset until $O$ years has passed.

\subsubsection{Personal Crises and Profit Opportunities}

There are two types of exogenous shocks to the owner of the house:

\begin{itemize}
\item \textit{Personal crises:} If the owner encounters an external shock
that is not market-related (such as losing his job, family- or job-related
relocation necessity, personal insolvency, defaulting on the loan, etc.), he
will immediately try to sell the house. This personal crisis could represent
a default on the loan, a foreclosure, and the transfer of ownership to the
bank, who would then put the house on the market. Including default and
foreclosure in the personal crises is essential for relating our price
evolution to that experienced in the recent subprime mortgage market crisis.

\item \textit{Market related profit opportunities:} After using the house
for $O$ years, the owner awaits an optimal market environment to post the
house for sale. While waiting for the optimal market environment, he still
faces the possibility of a personal crisis. We assume that the owner posts
the asset for sale whenever the interest rates fall below a certain
threshold, $\phi $. At this level, the owner expects to receive many offers
exceeding his reservation price.
\end{itemize}

As a result of these two exogenous shocks, the house will eventually be
offered for sale. Denoting the time to personal crisis by $\omega $ and the
time to a possible profit opportunity by $\pi $, then time to posting the
house for sale, $\upsilon $, is

\begin{equation}
\upsilon =\min (\omega ,\pi )\ \mbox{where }\pi =\inf (t\geq X:r(t)\leq \phi
).
\end{equation}%
Consequently, $\upsilon $ is the total occupation time, after which the
house is posted for sale with an initial list price, $L_{0}\geq R$. The
reservation price $R$ is private.

\subsubsection{The Sale Process with an OWT}

After the house is posted for sale, the seller sets an optimal waiting time,
OWT, which maximizes his expected utility. During the OWT, he collects
offers from prospective buyers and sells the asset immediately if he
receives an offer greater than the current list price. Everyone in the
market knows the list price, $L(T)$, but they do not know the seller's
reservation price, $R$. In our model, the list price is a function of $T$,
and gradually converges to $R$.\footnote{%
Note that $T$ represents the waiting time and $t$ is the actual time.} We
assume that $L(T)=R+(L_{0}-R)e^{-\zeta T}$ where $L_{0}$ is the initial list price and $%
\zeta $ is a positive real number.

Let $p_{1}(T_{1}),...,p_{k}(T_{k})$ be the offers received during waiting
time where $p_{k}(T_{k})$ is the $k$th offer received at time $T_{k}$. The
owner sells the house at time $T_{k}$ if $k$ satisfies:
\begin{equation}
p_{k}(T_{k})\geq L(T_{k})\ \text{and }p_{m}(T_{m})<L(T_{m})\ \forall m<k
\end{equation}%
If the owner does not receive any offers greater than the list price, but
receives offers exceeding $R$, then the owner sells the house to the highest
available offer, $p$, at the end of OWT. That is, if\ $\exists m$ such that $%
m\in \{1,2,...,n-1,n\}$, $p_{m}(T_{m})\geq R$ and $p_{m}(T_{m})$ is not
withdrawn, then $p=\max \left( {p_{k_{1}}(T_{k_{1}}),...,p_{k_{i}}(T_{k_{i}})%
}\right) $ where $i$ is the number of the available offers that are not
withdrawn until the end of OWT $\left( \text{Assume }i\geq 1\right) $.

If there is no offer exceeding $R$ at the end of the OWT, then the owner
must decrease his reservation and list price, and re-posts the house for
sale.

To complete the price evolution analysis, we use the framework introduced in
Section \ref{section:case2}. We assume that buyers make their offers at
random times according to a non-homogeneous Poisson distribution whose
intensities are defined by the demand function, $\lambda (t,r(t),L(T))$. At
these random times, they offer a price distributed uniformly arising from
their valuations of the asset. Buyers also withdraw their offers according
to a known random process. Different from the model in Section \ref%
{section:case2}, the arrival intensity of the offers equals $\lambda
(t,r(t),L(T))$. Interest rates in this analysis are no longer constant, but
stochastic. If we denote the expected utility function by $\mathcal{Y}%
(T,\lambda (t,T),\mu ,R,L(T),p_{\min },p_{\max },r(t),\gamma )$, then OWT is
defined as

\begin{definition}
$T^{\ast }=\arg {\max }\{T\geq 0:\mathcal{Y}(T,\lambda (t,T),\mu
,R,L(T),p_{\min },p_{\max },r(t),\gamma )\}$.
\end{definition}

\subsubsection{Updating Reservation and List Prices}

There are two different scenarios for updating the list and reservation
prices. If the seller sells the house before the end of the OWT, the new
reservation price for the next period equals the agreed sale price (as the
buyer of the house would not want to sell it for less). If the distribution
of the offers lies between $p_{\min }$ and $p_{\max }$, then the initial
post price $L_{0}$ equals $p_{\max }$ and $L(T)$ gradually decreases from $%
L_{0}$ to $R$ during the sale process.

If the seller does not succeed in selling the house, a new period starts for
this sale process with lower reservation and list prices. The seller chooses
a new reservation price which is between $R$ and $p_{\min }$ and sets $L_{0}$
to $R$. With these adjustments, he increases the probability of selling the
house in the next period.

\section{The Simulation}

This section presents the simulation results for the price evolution process.

\subsection{The Model Parameters}

We assume that interest rates evolve according to a Cox-Ingersoll-Ross
(1985) process:
\begin{equation*}
dr(t)=\kappa (\theta -r(t))dt+\sigma \sqrt{r(t)}dW(t).
\end{equation*}%
We assume that $\lambda (t,r(t),L(T))$ takes the following functional form
\begin{equation}
\lambda (t,r(t),L(T))=\frac{K_{1}}{r(t)}+\frac{K_{2}}{L(T)},
\end{equation}%
where $K_{1}$ and $K_{2}$ are constants. As either interest rates or the
list price increases, the probability of offers decline.

We first simulate occupation periods and exogenous shocks. The sale process
starts either with a personal shock or a profit opportunity. For each sale
process, we numerically calculate $T^{\ast }$ and model the offer arrival
times using the demand function as a non-homogeneous Poisson process. At the
arrival times, we generate independent offers using a uniform distribution
around $p_{\min }$ and $p_{\max }$. For the withdrawal times of the offers,
we use exponential distribution with parameter $\mu $. At the end of the
period, we update our reservation and list prices depending on the outcome
of the sale process. If there is a sale, the new reservation price equals $R$
and $L_{0}$ became $p_{\max }$. If there is no sale, then the new
reservation price equals $\frac{R+p_{\min }}{2}$ and $L_{0}$ takes the value
of $R$. A new period starts with these new parameters, current level of
interest rates and demand. The appendix includes the parameters used in the
simulation.

\subsection{Results}

Figure \ref{fig:OWT_summary} shows the evolution of the OWT over a four-year
period along with the demand intensity and interest rates for a single
realization. This graph differs from Figure \ref{fig:combined_lambda_mu_r}
in that it captures the evolution of all the parameters of the model. The
horizontal axis no longer represents the waiting time, but instead shows the
actual simulation time.

\begin{figure}[!t]
\begin{center}
\includegraphics[width=11cm]{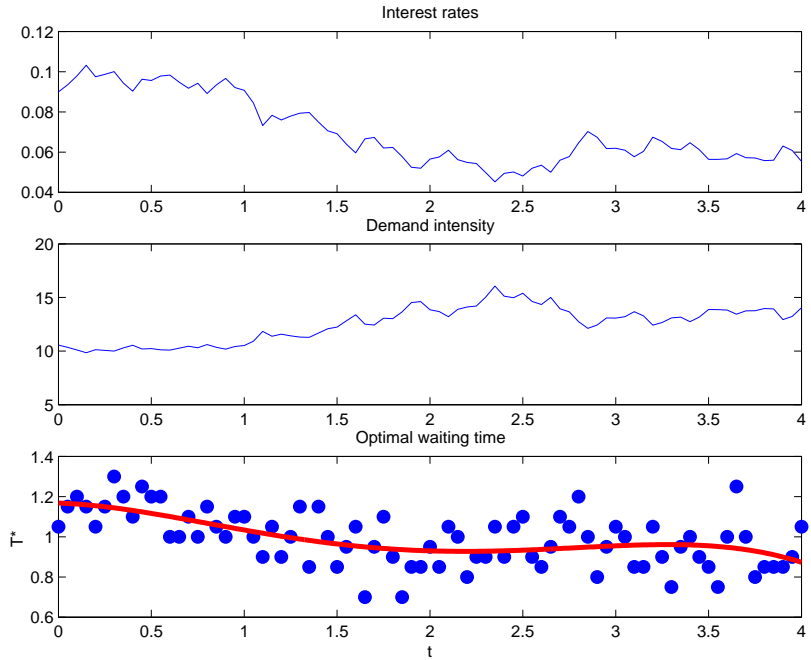} 
\end{center}
\caption{Evolution of optimal waiting time in a single realization of
interest rates and a stochastic demand function}
\label{fig:OWT_summary}
\end{figure}

This figure shows that as interest rates decrease, the demand intensity and
the OWT increase. With the increase in the demand intensity and the decrease
in the interest rates level, OWT decreases since, in these circumstances,
the seller gets more offers and his payoff is discounted by a lower factor.
This figure also demonstrates that changing market conditions have a broad
impact on the house sale process. In a low-interest rate and high-demand
environment, the seller keeps the house on the market for a shorter time
than he would in a high-interest and low-demand environment.

\begin{figure}[!t]
\begin{center}
\includegraphics[width=11cm]{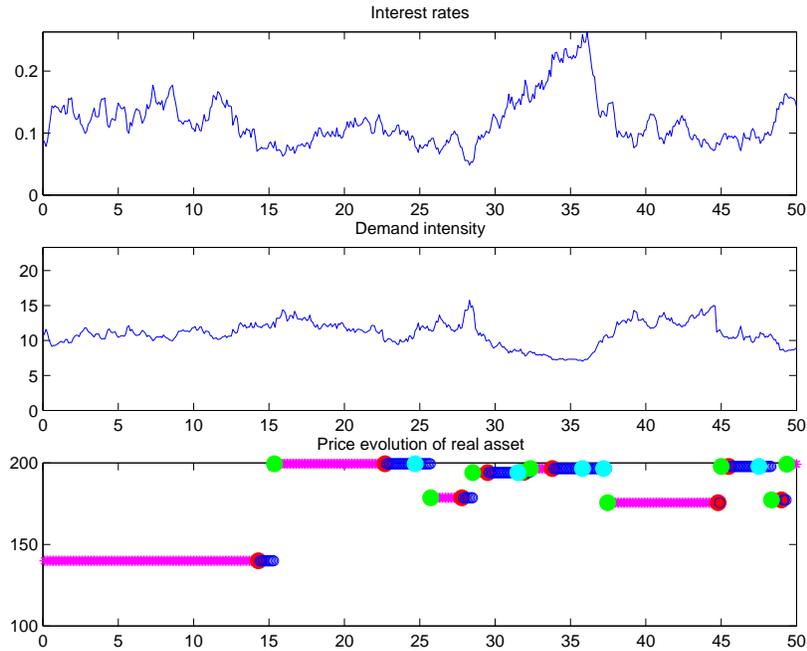} 
\end{center}
\caption{Single price evolution of a real asset in a realization of interest
rates and a stochastic demand function}
\label{fig:evolution}
\end{figure}

Figure \ref{fig:evolution} illustrates the house sale price along with the
corresponding functions of interest rates and demand intensities over a
fifty-year period. The price evolution in the third row shows the random
occupation periods, times of exogenous shocks due to personal crises or
profit opportunities. It also includes whether or not the price process
results in success with the appropriate labels shown in the legend of Figure %
\ref{fig:evolution_summary}. This figure extends the price evolution shown
in Figure \ref{fig:evolution}.

Figure \ref{fig:evolution_summary} studies the decision process of the
seller. The occupation period, decision of sale, time on the market, and
time of the sale are illustrated in the figure. Until the moment of sale,
all of these time periods are shown at the same price level which
constitutes the price that the seller paid at his initial purchase of the
house.

In Figure \ref{fig:evolution_summary}, pink circles represent the occupation
period in which the seller decides not to sell the asset. Pink circles are
always followed by a black or red dot which represents the decision of sale
due to a either a profit opportunity or a personal crisis, respectively. At
the time of sale decision, the seller sets an OWT, taking the current market
conditions into account. The waiting period is shown by the blue circles.
After this waiting time, a green or magenta dot follows, signifying the
event of sale and no-sale respectively. If there is a sale, a new period
starts with the occupation period of the new owner of the house.

\begin{figure}[!t]
\begin{center}
\includegraphics[width=11cm]{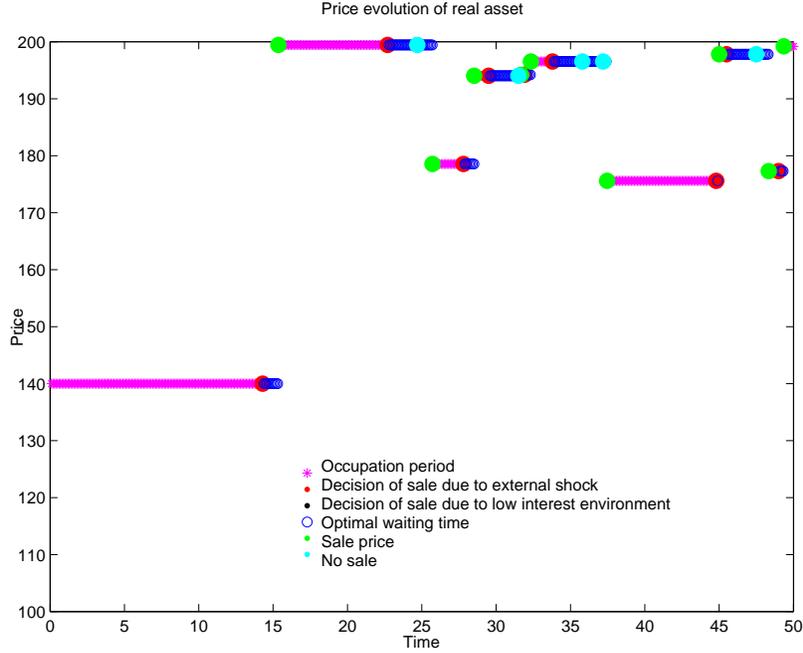} 
\end{center}
\caption{Single price evolution of a real asset with illustration of
subperiods and sale or no sale outcomes}
\label{fig:evolution_summary}
\end{figure}

\begin{figure}[!t]
\begin{center}
\includegraphics[width=11cm]{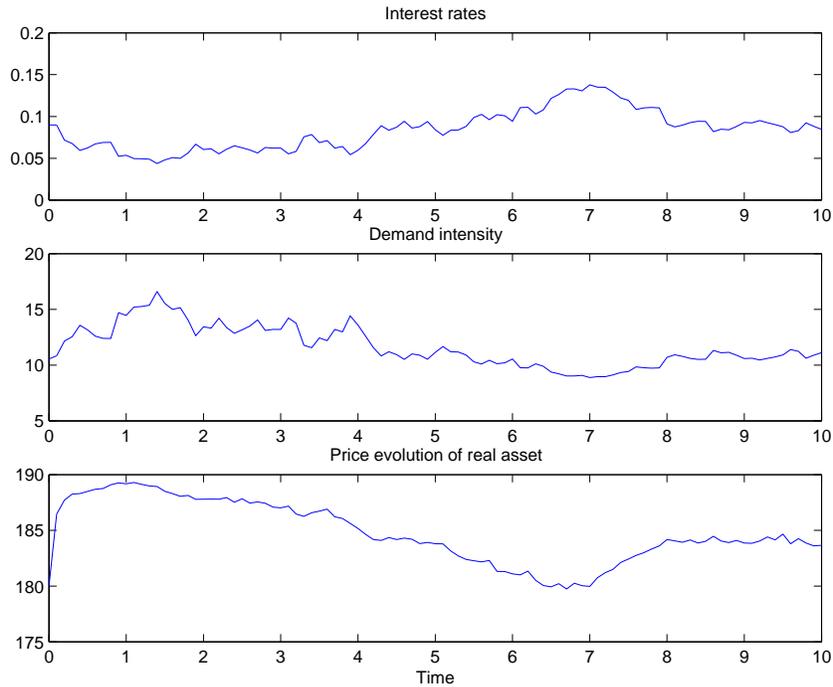} 
\end{center}
\caption{Expected price evolution of a real asset in a single realization of
interest rates and a stochastic demand function}
\label{fig:cont_price_evol}
\end{figure}


If the seller does not sell the house, he lowers the reservation and list
prices and sets a new OWT by using the new parameters and current market
conditions. Blue circles follow the magenta dot in the case of no-sale. In
this second sale attempt, he may still sell the asset for a price higher
than his initial reservation price, but this probability is less than in the
first attempt. Our results illustrate this conclusion, as in the figure,
four cases result in no-sales, only one of which is sold for a price higher
than the initial reservation price. The seller may not sell the house in
this attempt either. In this case, he will be further lower his reservation
and list prices. However, this probability is also lower: out of four cases
of no-sale, only one seller encountered two successive failed attempts.

Figure \ref{fig:cont_price_evol} graphs the expected price evolution of the
real asset without considering occupation periods or exogenous shocks. This
figure shows the mean sale price if the seller posts the house for sale at a
given time. This figure still assumes that the seller sets an optimal
waiting time and has his own reservation price. This figure supports the
conclusions from the single simulation. Note that the expected sale price
declines (increases) with a decrease (increase) in the demand and a rise
(fall) in interest rates.

Our simulation results show that the fluctuation in house prices are driven
by interest rates, demand, and time-on-the-market (TOM). In Figure \ref%
{fig:evolution}, interest rates increase between the 25th and 40th year, and
this increase results in two unsuccessful sale attempts in Figure \ref%
{fig:evolution_summary}. During this period, the reservation price of the
house is close to its maximum and, with the high-interest and low-demand
environment, the seller could not succeed in selling the house in two
trials. Eventually, the seller sells the asset in the third trial with a
significant discount compared to his reservation price.

\subsection{The Subprime Lending Crisis}

Our simulation provides some insights into the housing market crash observed
during the subprime lending crisis (for a review see Crouhy, Jarrow,
Turnbull \cite{CJT}). The record low borrowing rates during 1999-2004
increased housing affordability with the corresponding up-trend in home
prices as shown in Figure \ref{fig:rate_inv}. When borrowing rates increased
and home buyers' mortgage contracts began to reset to these higher rates,
low credit borrowers had a difficult time making their monthly payments.
They either defaulted or were forced to put their homes up for sale. This
coincides exactly with our modeling of a \textit{personal crisis}. Due to
this shock, the home owner had to sell his house in an unfavorable market:
high interest rates and low demand. As shown in our simulation results, the
expected price in such a market environment is low.

Another catalyst for subprime defaults were the higher reservation prices
locked in during the bull period. Although demand should have been low
during this period, borrowers still obtained financing - helped by the lower
standards of the mortgage originators. To support this view, note that
during 2002-2006, the Combined Loan to Value (CLTV) ratio increased. In the
subprime category, it rose to 88\% from 81\% and in Alt-A category, it rose
to 85\% from 73\%\footnote{%
Source: UBS Mortgage Research.}. When the subprime homes were put on the
market, they could not get any offers matching their list price. As we
discussed in our simulation, in a high-interest and low demand environment,
it is very difficult to sell a property with a high reservation price. Since
most of the subprime borrowers paid for their homes at record high values,
they had to lower their reservation prices during the sale process. This
result is similar to our successive no-sale events when the reservation
price of the owner is close to $p_{\max }$.

\section{Conclusion}

This paper proposes a new model to describe the evolution of housing prices.
We include exogenous shocks to the real estate market that may be the result
of a profit opportunity or personal crisis. We investigate the sale process
by introducing a new model for analyzing the time-on-the-market (TOM) with a
new construct, the optimal waiting time (OWT). We study the comparative
statistics of the OWT with respect to the model's parameters such as the
arrival intensity, the offer withdrawal (cancelation) rates, and interest
rates. We specifically look at the pairwise impact of these parameters and
how they affect the resulting OWT in the different regions of the surface.
We incorporate our theory of the OWT into a price evolution analysis that
also includes an occupation period, a personal crisis or profit-taking
opportunity, and the deterministic updating of model parameters with the
occurrence of the sale. Our model specifically considers the occurrence of
sale, no sale conditions, and how the seller responds to the no sale
scenario. Our simulation results show that it becomes more difficult to sell
a house in a high-interest, low-demand environment and that these conditions
may require the seller to sell the asset below his initial reservation
price. This results in a dynamic, time-dependent, and stochastic house price
process which provides some insights into the recent subprime credit crisis.

\clearpage%

\break

\appendix

\section{Appendix: Auxiliary Model}

We will derive the auxiliary model in this appendix. Buyers make offers with
$\exp (\lambda)$ and their offers are distributed uniformly with $U(p_{\min
},p_{\max })$. The seller's reservation price is $p_{\min }$. After making
an offer, a buyer may withdraw his offer with distribution $\exp (\mu )$.
Interest rate is constant and equals $r$. The seller wants to maximize
expected payoff with respect to waiting time, $T$. Let $N(T)$ be the number
of offers received by time $T$ and $\xi_{i}$ be the offer from buyer $i$, $%
B_{i}$, at the arrival time, $A_{i}$, and $\xi_{(i)}$ be the $i$th minimum
offer received by the seller. With these assumptions, the discounted
expected payoff, $u(.)$, is a function of $T$, $\lambda$, $\mu$, $p_{\min}$,
$p_{\max}$, and $r$.\newline

\begin{eqnarray*}
u(.) &=&\mathbb{E}\left[ \mathbb{E}\left[X\vert N(T)=n\right] \right] \\
&=&\sum_{n=0}^{\infty }\mathbb{E}\Big[\xi_{(n)}e^{-rT}1 \! \! 1_{[B_{n}\text{
is still interested}]}+\xi_{(n-1)}e^{-rT}1 \! \! 1_{[B_{n-1}\text{ is still
interested and }\xi_{(n)}\text{ is withdrawn}]} \\
&&\text{ \ \ \ \ }\ldots +\xi_{(1)}e^{-rT}1 \! \! 1_{[B_{1}\text{ is still
interested and }\xi_{(n)}...\xi_{(2)}\text{ are withdrawn}]}|N(T)=n\Big]
\mathbb{P}\left(N(T)=n\right) \\
&=&\sum_{n=0}^{\infty }\sum_{i=1}^{n}\mathbb{E}\left[ \xi_{(n-i+1)}e^{-rT}1
\! \! 1_{[B_{n-i+1}\text{ is still interested and }\xi_{(n)}...\xi_{(n-i+2)}%
\text{ are withdrawn}]}|N(T)=n\right] \times \\
&&\text{ \ \ \ \ \ } \mathbb{P}\left(N(T)=n\right) \\
&=&\sum_{n=0}^{\infty }\sum_{i=1}^{n}\mathbb{E}\left[ \xi_{(n-i+1)}|N(T)=n%
\right] \times \mathbb{E}\left[ e^{-rT}1 \! \! 1_{[B_{n-i+1}\text{ is still
interested}]}|N(T)=n\right] \times \\
&&\text{ \ \ \ \ \ }\mathbb{E}\left[ 1 \! \! 1_{[\xi_{(n)}...\xi_{(n-i+2)}%
\text{ are withdrawn}]}|N(T)=n\right] \mathbb{P}\left(N(T)=n\right).
\end{eqnarray*}

Each of the components of the sum is as follows
\begin{eqnarray*}
\mathbb{E}\left[ \xi_{(i)}|N(T)=n\right] &=&p_{\min }+\frac{(p_{\max
}-p_{\min })i}{n+1}, \\
\mathbb{E}\left[ 1 \! \! 1_{[B_{i}\text{ is still interested}]}|N(T)=n\right]
&=&\int_{0}^{T}\frac{1}{T}e^{-\mu x}dx=\frac{1}{\mu T}(1-e^{-\mu T}):=1-f(T),
\\
\mathbb{E}\left[ 1 \! \! 1_{[\xi_{(n)}...\xi_{(n-i+2)}\text{ are withdrawn}%
]}|N(T)=n\right] &=&[f(T)]^{i-1}, \\
\mathbb{P}\left( N(T)=n\right) &=&\frac{(\lambda T)^{n}e^{-\lambda T}}{n!}.
\end{eqnarray*}

Using these components, expected payoff becomes
\begin{eqnarray*}
u(.) &=&\sum_{n=0}^{\infty }\sum_{i=1}^{n}\left( p_{\min }+\frac{(p_{\max
}-p_{\min })(n-i+1)}{n+1}\right) e^{-rT}(1-f(T))(f(T))^{i-1}\left( \frac{%
(\lambda T)^{n}e^{-\lambda T}}{n!}\right) \\
&=&g(T)\sum_{n=0}^{\infty }\frac{(\lambda T)^{n}}{n!}\sum_{i=1}^{n}\left(
p_{\max }(f(T))^{i-1}-\frac{(p_{\max }-p_{\min })}{n+1}i(f(T))^{i-1}\right),
\end{eqnarray*}%
{\ where }$g(T)=(1-f(T))e^{-rT}e^{-\lambda T}.$ We use the following facts
in the final computation.
\begin{eqnarray*}
\sum_{i=1}^{n}ix^{i-1} &=&\frac{d(\sum_{i=1}^{n}x^{i})}{dx}=\frac{d(\frac{%
1-x^{n+1}}{1-x})}{dx}=\frac{nx^{n+1}-(n+1)x^{n}+1}{(1-x)^{2}} \\
\mbox{and} \hspace{1cm} & & \\
\sum_{i=1}^{n}x^{i-1} &=&\frac{1-x^{n}}{1-x}.
\end{eqnarray*}

Finally, the discounted expected payoff equals
\begin{eqnarray*}
u(T,\lambda ,\mu ,p_{\min },p_{\max },r) &=&g(T)\sum_{n=0}^{\infty }\frac{%
(\lambda T)^{n}}{n!}\Big[\frac{p_{\max }-p_{\min }}{(n+1)(1-f(T))^{2}}\left(
nf(T)^{n+1}-(n+1)f(T)^{n}+1\right) \\
&&\text{ \ \ \ \ \ \ \ \ \ }+p_{\max }\frac{1-f(T)^{n}}{1-f(T)}\Big] \\
&=&\frac{-g(T)(p_{\max }-p_{\min })}{(1-f(T))^{2}}(f(T)e^{\lambda Tf(T)}-%
\frac{1}{\lambda T}(e^{\lambda Tf(T)}-1)-e^{\lambda Tf(T)} \\
&&\text{ \ \ \ \ \ \ \ \ \ }+\frac{1}{\lambda T}(e^{\lambda T}-1))+\frac{%
p_{\max }g(T)(e^{\lambda T}-e^{\lambda Tf(T)})}{1-f(T)},
\end{eqnarray*}%
where $f(T)=1-\frac{1}{\mu T}(1-e^{-\mu T})$ and $g(T)=(1-f(T))e^{-rT}e^{-%
\lambda T}$.

\newpage

\section{Appendix: Theorem \protect\ref{Thm: Payoff with changing Sale price}%
}

We first establish some auxiliary results that will be used while proving
Theorem \ref{Thm: Payoff with changing Sale price}. Lemma \ref{Lemma: phi(t)}
gives us the formula for the conditional probability, which is conditioned
on $\mathcal{F}_t$ and $N(t) = n$, that an offer value is smaller than the
announced list price at its arrival time.

\begin{lemma}
\label{Lemma: phi(t)} Let $\varphi(t) = \mathbb{P}\{\xi_1 < L(A_1) |
\mathcal{F}_t, N(t)=n\}$ for $n \geq 1$. Then,
\begin{eqnarray}
\varphi(t) = \frac{1}{\Lambda(t)} \int_0^t \lambda(a) F_\xi\big(L(a)\big)
\mathrm{d}a.  \notag
\end{eqnarray}
\end{lemma}

\begin{proof}
Given $A_1 = a$, the event $\{\xi_1 < L(A_1)\}$ is independent of $\mathcal{F%
}_t$ and $N(t)$. Thus,
\begin{eqnarray}
\varphi(t) &=& \int_0^t f_{A|\{\mathcal{F}_t, N(t)=n\}}(a) \mathbb{P}\{\xi_1
< L(a) \} \mathrm{d}a  \notag \\
&=& \frac{1}{\Lambda(t)} \int_0^t \lambda(a) F_\xi \big( L(a) \big) \mathrm{d%
}a.  \notag
\end{eqnarray}
\end{proof}

The following lemma provides the formula for the conditional probability,
which is conditioned on $\mathcal{F}_t$ and $N(t) = n$, that none of the
offers arrived in the time interval $[0,t]$ is greater than the announced
sale at their arrival times.

\begin{lemma}
\label{Lemma: density of beta} $\mathbb{P}\{\beta>t | \mathcal{F}_t, N(t) =
n\} = \varphi(t)^n$ for $n \geq 1$.
\end{lemma}

\begin{proof}
\begin{eqnarray}
\mathbb{P}\{\beta>t | \mathcal{F}_t, N(t) = n\} &=& \mathbb{P}\Big( %
\bigcap_{i=0}^n \{\xi_i < L(A_i)\} | \mathcal{F}_t, N(t) = n \Big)  \notag \\
&=& \mathbb{P}\{ \xi_1 < L(A_1) | \mathcal{F}_t, N(t) = n \}^n =
\varphi(t)^n.  \notag
\end{eqnarray}
\end{proof}

Now, we calculate the conditional density of $\beta$, $f_{\beta | \{\mathcal{%
F}_t, N(t)=n, \beta \leq t\}}(s)$, conditioned on $\mathcal{F}_t$, $N(t)=n$,
and $\beta \leq t$.

\begin{lemma}
$f_{\beta | \{\mathcal{F}_t, N(t)=n, \beta \leq t\}}(s) = \frac{\lambda(s)}{%
\Lambda(t)}$.
\end{lemma}

\begin{proof}
Let $i(\beta)$ be the index of the offer at time $\beta$. Then,
\begin{eqnarray}
\mathbb{P}\{\beta \leq s | \mathcal{F}_t, N(t)=n, \beta \leq t\} = \mathbb{P}%
\{A_{i(\beta)} \leq s | \mathcal{F}_t, N(t)=n, \beta \leq t\}.  \notag
\end{eqnarray}
Since $\beta \leq t$, we know that $i(\beta) \leq n$. Given $N(t)=n$ and $%
\mathcal{F}_t$, all offer arrival times are independently distributed over $%
[0,t]$ according to density $\frac{\lambda(s)}{\Lambda(t)}$. Thus,
\begin{eqnarray}
\mathbb{P}\{\beta \leq s | \mathcal{F}_t, N(t)=n, \beta \leq t\} = \frac{%
\Lambda(s)}{\Lambda(t)}.  \notag
\end{eqnarray}
Thus, $f_{\beta | \{\mathcal{F}_t, N(t)=n, \beta \leq t\}}(s) = \frac{%
\lambda(s)}{\Lambda(t)}$.
\end{proof}

In Lemma \ref{Lemma: offer probability}, we derive the conditional
probability, conditioned on $\mathcal{F}_t$ and $N(t)=n$, of an offer, which
is not withdrawn up to time $t$ and greater than the reservation price but
not exceeding the list price at its arrival time, to be greater than a
positive real number $y$.

\begin{lemma}
\label{Lemma: offer probability} Let $\psi(t,y) = \mathbb{P}\{\xi_1 1 \! \!
1_{\{R \leq \xi_1 < L(A_1)\}} 1 \! \! 1_{\{\tau_1 \geq t - A_1 \}} > y |
\mathcal{F}_t, N(t) = n\}$ for $n \geq 1$. Then,
\begin{eqnarray}
\psi(t,y) = \frac{1}{\Lambda(t)} \int_0^t \lambda(a) \Big(1-F_\tau(t-a)\Big)%
\Big(F_\xi\big(L(a)\vee y\big) - F_\xi\big(R \vee y\big)\Big) \mathrm{d}a.
\end{eqnarray}
\end{lemma}

\begin{proof}
Conditioned on $A_1 = a$, events $\{\tau_1 \geq t-A_1\}$ and $\{R \leq \xi_1
< L(A_1)\}$ are independent of each other as well as being independent of $%
\mathcal{F}_t$ and $N(t)$. Thus,
\begin{eqnarray}
\psi(t,y) &=& \frac{1}{\Lambda(t)} \int_0^a \lambda(a) \mathbb{P}\{\tau_1
\geq t-a\} \mathbb{P}\{\xi_1 1 \! \! 1_{\{R \leq \xi_1 < L(a)\}} > y\}
\notag \\
&=& \frac{1}{\Lambda(t)} \int_0^t \lambda(a) \Big(1-F_\tau(t-a)\Big)\Big(%
F_\xi\big(L(a)\vee y\big) - F_\xi\big(R \vee y\big)\Big) \mathrm{d}a.  \notag
\end{eqnarray}
\end{proof}

Let $p_n(t) = \mathbb{P}\{N(t) = n | \mathcal{F}_t \}$. The following
summation formula will also be used in the proof of Theorem \ref{Thm: Payoff
with changing Sale price}, and gives us the conditional moment generating
function of a Poisson process with stochastic intensity.

\begin{lemma}
\label{Lemma: Sum Formulas} For $q>0$, the following holds.
\begin{eqnarray}
\sum_{n=0}^\infty p_n(t) q^n &=& \exp\big(\Lambda(t) (q-1) \big).  \notag
\end{eqnarray}
\end{lemma}

We now start proving Theorem \ref{Thm: Payoff with changing Sale price}. Let
$X_1(t)$ be the first term in (\ref{Eqn: Discounted Payoff}), and $X_2(t)$
be the second term in (\ref{Eqn: Discounted Payoff}). Let us define $P_1(t |
\mathcal{F}_t) = \mathbb{E}[X_1(t) | \mathcal{F}_t]$, and $P_2(t | \mathcal{F%
}_t) = \mathbb{E}[X_2(t) | \mathcal{F}_t]$. We first calculate $P_1(t |
\mathcal{F}_t)$.
\begin{eqnarray*}
P_1(t | \mathcal{F}_t) = \sum_{n=0}^\infty p_n(t) \mathbb{E}\big[X_1(t) |
\mathcal{F}_t, N(t) = n\big].  \notag
\end{eqnarray*}

Put $M(t) = \max_{0 \leq i \leq N(t)} \xi_i 1 \! \! 1_{\{ \xi_i \geq R\}}1
\! \! 1_{\{\tau_i \geq t - A_i\}}$. Then,
\begin{eqnarray}
\mathbb{E}[X_1(t) | \mathcal{F}_t, N(t) = n] = \mathrm{e}^{-\int_0^t r(s)
\mathrm{d}s} \mathbb{E}\big[M(t) 1 \! \! 1_{ \{\beta>t\}} | \mathcal{F}_t,
N(t) = n \big].  \label{Eqn: Conditional Expectation X1}
\end{eqnarray}
The expectation in (\ref{Eqn: Conditional Expectation X1}) can be calculated
by conditioning on the event $\{\beta>t\}$.
\begin{eqnarray*}
\mathbb{E}\big[M(t) 1 \! \! 1_{ \{\beta>t\}} | \mathcal{F}_t, N(t) = n \big] %
&=& \mathbb{P}\{\beta>t|\mathcal{F}_t, N(t)=n\} \mathbb{E}\big[M(t) |
\mathcal{F}_t, N(t)=n, \beta>t \big]  \notag \\
&=& \varphi(t)^n \mathbb{E}\big[M(t) | \mathcal{F}_t, N(t)=n, \beta>t \big].
\notag
\end{eqnarray*}

Since $M(t)$ is positive, its conditional expectation can be calculated by
integrating $\mathbb{P}\{M(t)>y | \mathcal{F}_t, N(t)=n, \beta>t\}$ with
respect to $y$ over $[0,\infty]$. Let us calculate $\mathbb{P}\{M(t)>y |
\mathcal{F}_t, N(t)=n, \beta>t\}$.
\begin{eqnarray*}
\lefteqn{\prob\{M(t)>y | \mathcal{F}_t, N(t)=n, \beta>t\}} \hspace{14cm}
\notag \\
\lefteqn{= 1- \Big( 1 - \prob\big\{\xi_1 \Indic_{ \{\xi_1 \geq R\}}
\Indic_{\{\tau_1 \geq t - A_1\}} > y | \mathcal{F}_t, N(t)=n, \xi_1 <
L(A_1)\big\} \Big)^n} \hspace{12cm}  \notag \\
\lefteqn{=1-\Big(1-\frac{\prob\{\xi_1 \Indic_{ \{R \leq \xi_1 \leq L(A_1)\}}
\Indic_{\{\tau_1 \geq t - A_1\}} > y | \mathcal{F}_t, N(t)=n \}}{%
\prob\{\xi_1 < L(A_1)|\mathcal{F}_t, N(t)=n\}} \Big)^n } \hspace{12cm}
\notag \\
\lefteqn{=1 - \Big(1-\frac{\psi(t,y)}{\varphi(t)} \Big)^n } \hspace{12cm}
\notag
\end{eqnarray*}

Integrating $\mathbb{P}\{M(t)>y | \mathcal{F}_t, N(t)=n, \beta>t\}$ over $y$%
, we obtain $\mathbb{E}[M(t)| \mathcal{F}_t, N(t)=n, \beta>t]$.

\begin{eqnarray}
\mathbb{E}[M(t)| \mathcal{F}_t, N(t)=n, \beta>t] &=& \int_0^\infty \mathbb{P}%
\{M(t)>y | \mathcal{F}_t, N(t)=n, \beta>t\} \mathrm{d}y  \notag \\
&=& \int_0^\infty \Big(1-\big(1-\frac{\psi(t,y)}{\varphi(t)} \big)^n \Big)
\mathrm{d}y.  \label{Eqn: Expectation M(t)}
\end{eqnarray}
Noting that $\psi(t,y) = 0 $ a.s. when $y > L_0$, we can further simplify (%
\ref{Eqn: Expectation M(t)}) to
\begin{eqnarray*}
\mathbb{E}[M(t)| \mathcal{F}_t, N(t)=n, \beta>t] = L_0 - \int_0^{L_0} \big(1-%
\frac{\psi(t,y)}{\varphi(t)} \big)^n \mathrm{d}y.
\end{eqnarray*}

As a result, $\mathbb{E}[X_1(t)| \mathcal{F}_t, N(t) = n]$ is equal to
\begin{eqnarray*}
\mathbb{E}[X_1(t)| \mathcal{F}_t, N(t) = n] = \mathrm{e}^{-\int_0^t r(s)%
\mathrm{d}s} \big(L_0 \varphi(t)^n - \int_0^{L_0} \big(\varphi(t) -
\psi(t,y) \big)^n \mathrm{d}y\big).
\end{eqnarray*}
Now, we average $\mathbb{E}[X_1(t)| \mathcal{F}_t, N(t) = n]$ over $N(t)$.
By using Fubini's theorem and Lemma \ref{Lemma: Sum Formulas}, we obtain
\begin{eqnarray*}
P_1(t | \mathcal{F}_t) = \mathrm{e}^{-\int_0^t r(s)\mathrm{d}s }\mathrm{e}%
^{\Lambda(t)(\varphi(t)-1)} \Big(L_0 - \int_0^{L_0} \mathrm{e}^{-\Lambda(t)
\psi(t,y)} \mathrm{d}y \Big).
\end{eqnarray*}

Let us now calculate $P_2(t|\mathcal{F}_t)$. It is equal to
\begin{eqnarray*}
P_2(t|\mathcal{F}_t) = \sum_{n=0}^\infty p_n(t) \mathbb{E}[X_2(t)|\mathcal{F}%
_t, N(t)=n] .  \notag
\end{eqnarray*}
where $\mathbb{E}[X_2(t) | \mathcal{F}_t, N(t)=n]$ is calculated as
\begin{eqnarray*}
\mathbb{E}[X_2(t) | \mathcal{F}_t, N(t)=n] &=& \mathbb{E}\big[\mathrm{e}%
^{-\int_0^\beta r(s)\mathrm{d}s} \xi_{i(\beta)} 1 \! \! 1_{\{\beta \leq t\}}
| \mathcal{F}_t, N(t)=n \big]  \notag \\
&=& \mathbb{P}\{\beta \leq t | \mathcal{F}_t, N(t)=n\} \mathbb{E}\big[%
\mathrm{e}^{-\int_0^\beta r(s)\mathrm{d}s} \xi_{i(\beta)} | \mathcal{F}_t,
N(t)=n, \beta \leq t \big]  \notag \\
&=& \big(1-\varphi(t)^n\big) \int_0^t f_{\beta | \{\mathcal{F}_t, N(t)=n,
\beta \leq t\}}(a) \mathrm{e}^{-\int_0^a r(s)\mathrm{d}s} \mathbb{E}\big[%
\xi_{i(a)}\big].  \notag
\end{eqnarray*}
The last equality follows from the fact that the magnitude of $%
\xi_{i(\beta)} $ is independent of $\mathcal{F}_t$ and $N(t)$, and only
depends on $a$ given the event $\{\beta=a\}$. Note also that
\begin{eqnarray*}
\mathbb{E}[\xi_{i(a)}] &=& \mathbb{E} [\xi_1 | \xi_1 \geq L(a) ]  \notag \\
&=& \frac{\int_{L(a)}^\infty x f_\xi(x) \mathrm{d}x}{1-F_\xi\big(L(a)\big)}
\end{eqnarray*}

Thus,
\begin{eqnarray*}
\mathbb{E}[X_2(t) | \mathcal{F}_t, N(t)=n] &=& \big(1-\varphi(t)^n\big) %
\int_0^t \int_{L(a)}^\infty f_{\beta | \{\mathcal{F}_t, N(t)=n, \beta \leq
t\}}(a) \frac{\exp(-\int_0^a r(s)\mathrm{d}s)}{1-F_\xi\big(L(a)\big)}x
f_\xi(x)\mathrm{d}x\mathrm{d}a  \notag \\
&=& \big(1-\varphi(t)^n\big) \int_0^t \int_{L(a)}^\infty \frac{\lambda(a)\exp%
\big(-\int_0^a r(s)\mathrm{d}s\big)}{\Lambda(t)\big(1-F_\xi\big(L(a)\big)%
\big)} f_\xi(x) x \mathrm{d}x \mathrm{d}a.  \notag
\end{eqnarray*}
$P_2(t|\mathcal{F}_t)$ is obtained by averaging $\mathbb{E}[X_2(t) |
\mathcal{F}_t, N(t)=n]$ over $N(t)$. By using Lemma \ref{Lemma: Sum Formulas}%
, we obtain
\begin{eqnarray*}
P_2(t|\mathcal{F}_t) = \big(1-\mathrm{e}^{\Lambda(t)(\varphi(t)-1)}\big) %
\int_0^t \int_{L(a)}^\infty \frac{\lambda(a)\exp\big(-\int_0^a r(s)\mathrm{d}%
s\big)}{\Lambda(t)\big(1-F_\xi\big(L(a)\big)\big)} f_\xi(x) x \mathrm{d}x
\mathrm{d}a.
\end{eqnarray*}
This completes the proof since $P(t|\mathcal{F}_t) = P_1(t|\mathcal{F}_t) +
P_2(t|\mathcal{F}_t)$.

{\ \bigskip \newpage }

\section{Appendix: Parameter Assumptions}

\begin{table}[htc]
\begin{center}
\begin{tabular}{|l|l|}
\hline
\textbf{Parameter}&\textbf{Value}\\
\hline

Arrival intensity ($\lambda$) & 5\\
Withdrawal intensity ($\mu$) & 5\\
Interest rate ($r$) & 0.1\\
Reservation price ($R$) & 140\\
List price ($L$) & 180\\
Waiting averseness ($\gamma$) & 0.1\\
$p_{\min }$ & 100\\
$p_{\max }$ & 200\\
\hline
\end{tabular}
\caption{Default parameter values in OWT analysis in Section 2}

\end{center}
\end{table}

\begin{table}[htc]
\begin{center}
\begin{tabular}{|l|l|}
\hline
\textbf{Parameter}&\textbf{Value}\\
\hline
Withdrawal intensity ($\mu$) & 10\\
Occupation period ($O$) & $U(4,6)$\\
Rate of personal crisis & $Exp(10)$\\
Initial reservation price($R_0$) & 140\\
Initial list price ($L_0$) & 200\\
Waiting averseness ($\gamma$) & 0.8\\
Interest rate threshold ($\phi$) & 0.06\\
$p_{\min }$ & 100\\
$p_{\max }$ & 200\\
$K_1$ & 0.5\\
$K_2$ & 1000\\
$\theta$ & 0.1\\
$\sigma$ & 0.08\\
$\kappa$ & 0.25\\
$r_0$ & 0.09\\
\hline
\end{tabular}
\caption{Default parameter values in the simulation}
\end{center}
\end{table}

\end{document}